\documentclass{article}
\usepackage[dvipsnames]{xcolor}
\usepackage{pgfplots}
\usepackage{caption}
\usepackage{subcaption}
\usepackage{multicol}
\usepackage{enumerate}
\usepackage{etoolbox}
\usepackage{amsmath,latexsym,amssymb,verbatim,graphicx,fancyhdr}

\usepackage{amsthm}
\usepackage{booktabs}
\usepackage{bbm}
\usepackage{tikz}

\usepackage{multirow}
\usepackage{multicol}
\usepackage{makecell}
\usepackage{hhline}
\usepackage{footnote}
\makesavenoteenv{table}
\usepackage{ytableau}
\usetikzlibrary{matrix}
\usetikzlibrary{arrows,positioning} 

\usepackage{bm}
\usepackage{nicefrac}
\usepackage{mdwlist}
\usepackage{fullpage}

\usepackage[toc,title,page]{appendix}

\usepackage[english]{babel}
\addto\extrasenglish{%
	\def\appendixautorefname{Appendix}%
}
\usepackage[utf8]{inputenc}
\usepackage[T1]{fontenc}
\usepackage{lmodern}
\usepackage{amsmath,amsthm,mathdots,amssymb,bm,bbm,mathrsfs,mathtools}
\usepackage{tikz,pgfplots}\pgfplotsset{compat=1.16}
\usepackage[section]{placeins}
\usepackage{thm-restate}
\newtheorem{thm}{Theorem}[section]

\newtheorem{lem}{Lemma}[section]

\newtheorem{defn}{Definition}[section]

\usepackage{comment}

\usepackage{mleftright}\mleftright 

\usepackage{enumitem}
\usepackage{microtype,color} 
\usepackage[colorlinks=true, urlcolor=violet, linkcolor=blue, citecolor=OliveGreen, hyperindex=true, linktocpage=true, pagebackref=true, draft=false]{hyperref}

\usepackage{cleveref}

\newcommand{\CL}{\mathcal{L}}

\newcommand{\vA}{\bm{A}}

\newcommand{\vB}{\bm{B}}

\newcommand{\vH}{\bm{H}}
\newcommand{\vI}{\bm{I}}

\newcommand{\vL}{\bm{L}}

\newcommand{\vX}{\bm{X}}

\newcommand{\vZ}{\bm{Z}}
\newcommand{\vsigma}{\bm{ \sigma}}

\newcommand{\vrho}{\bm{ \rho}}

\newcommand{\vertiii}[1]{{\left\vert\kern-0.25ex\left\vert\kern-0.25ex\left\vert #1 \right\vert\kern-0.25ex\right\vert\kern-0.25ex\right\vert}}
\newcommand{\norm}[1]{\Vert {#1} \Vert}

\newcommand{\e}{\mathrm{e}}

\newcommand{\rd}{\mathrm{d}}
\newcommand*{\tr}{\mathrm{Tr}}
\newcommand*{\poly}{\mathrm{poly}}
\newcommand*{\polylog}{\mathrm{polylog}}

\newcommand{\Lword}[1]{\text{Lindbladian}}

\newcommand{\floor}[1]{\left\lfloor{#1}\right\rfloor}

\DeclarePairedDelimiterX{\braket}[1]{\langle}{\rangle}{#1}
\DeclarePairedDelimiterX\ketbra[2]{| }{|}{#1 \delimsize\rangle\!\delimsize\langle #2}	
\DeclarePairedDelimiterX\dotp[2]{\langle}{\rangle}{#1, #2}

\DeclareMathAlphabet{\dutchcal}{U}{dutchcal}{m}{n}
\SetMathAlphabet{\dutchcal}{bold}{U}{dutchcal}{b}{n}
\DeclareMathAlphabet{\dutchbcal} {U}{dutchcal}{b}{n}

\makeatletter
\DeclareRobustCommand*{\pmzerodot}{%
	\nfss@text{%
		\sbox0{$\vcenter{}$}
		\sbox2{0}%
		\sbox4{0\/}%
		\ooalign{%
			0\cr
			\hidewidth
			\kern\dimexpr\wd4-\wd2\relax 
			\raise\dimexpr(\ht2-\dp2)/2-\ht0\relax\hbox{%
				\if b\expandafter\@car\f@series\@nil\relax
				\mathversion{bold}%
				\fi
				$\cdot\m@th$%
			}%
			\hidewidth
			\cr
			\vphantom{0}
		}%
	}%
}

\usepackage{mmacells}

\mmaSet{
	morefv={gobble=2},
	linklocaluri=mma/symbol/definition:#1,
	morecellgraphics={yoffset=0mm},
	leftmargin=11.5mm,
	labelsep=1mm,
}

\usepackage{ifdraft}
\ifdraft{
	\newcommand{\authnote}[3]{{\color{#3} {\bf  #1:} #2}}	
}{
	\newcommand{\authnote}[3]{}
}

\usetikzlibrary{quantikz2}

\makeatletter
\def\l@subsection#1#2{}
\def\l@subsubsection#1#2{}
\makeatother

\usepackage{authblk}
\begin{document}
\renewcommand{\appendixautorefname}{Appendix}
\renewcommand{\chapterautorefname}{Chapter}
\renewcommand{\sectionautorefname}{Section}
\renewcommand{\subsubsectionautorefname}{Section}

\title{Mixing time of quantum Gibbs sampling \\ for random sparse {H}amiltonians}
\date{\today}

\author{Akshar Ramkumar}
\author{Mehdi Soleimanifar}
\affil{California Institute of Technology}

\maketitle
\begin{abstract}
Providing evidence that quantum computers can efficiently prepare low-energy or thermal states of physically relevant interacting quantum systems is a major challenge in quantum information science. 
A newly developed quantum Gibbs sampling algorithm \cite{chen2023efficient} provides an efficient simulation of the detailed-balanced dissipative dynamics of non-commutative quantum systems. 
The running time of this algorithm depends on the mixing time of the corresponding quantum Markov chain, which has not been rigorously bounded except in the high-temperature regime. 
In this work, we establish a $\polylog(n)$ upper bound on its mixing time for various families of random $n \times n$ sparse Hamiltonians at any constant temperature.
We further analyze how the choice of the jump operators for the algorithm and the spectral properties of these sparse Hamiltonians influence the mixing time. 
Our result places this method for Gibbs sampling on par with other efficient algorithms for preparing low-energy states of quantumly easy Hamiltonians. 
\end{abstract}

\section{Introduction}
\label{sec:introduction}
\noindent
One of the main anticipated applications of quantum computers is the simulation and characterization of quantum systems in condensed matter physics \cite{Wecker2015QCstronglycorrelated}, quantum chemistry \cite{McArdle2020QuantumChemistry}, and high-energy physics \cite{preskill2018simulating, bauer2023SimulatingHighEnergy}.
The problem of simulating the dynamics (time evolution) of an interacting quantum system under a local or sparse Hamiltonian $\vH$ has largely been addressed, with efficient algorithms \cite{haah2018QAlgSimLatticeHam, low2016HamSimQSignProc, berry2014HamSimTaylor, low2016HamSimQubitization, gilyen2018QSingValTransf} that scale well with the number of particles, simulation time, and required precision.
However, the ability of quantum computers to evaluate the static features of quantum systems, such as their ground state or thermal properties, is less understood.

In this work, we focus on preparing the Gibbs (thermal) state $\vrho_\beta = \frac{e^{-\beta \vH}}{\tr(e^{-\beta \vH})}$ of a quantum system, which represents the equilibrium state when the system is in contact with a thermal bath at a fixed temperature $\beta^{-1}$. This computational problem, known as Gibbs sampling or ``cooling,'' is valuable not only for simulating thermodynamic properties but also as a subroutine in quantum algorithms for optimization and learning \cite{brandao2016QSDPSpeedup, apeldoorn2018ImprovedQSDPSolving, brandao2017QSDPSpeedupsLearning}.
However, to prepare the Gibbs state, quantum computers face challenges.  
In general, it is not believed that estimating the low-temperature properties of quantum systems can be solved efficiently by a quantum computer in the worst-case \cite{complexity_local_H}. 
Fortunately, it has been hypothesized that this worst-case hardness of finding low-temperature states implied by arguments from complexity theory is due to pathological Hamiltonians, which are not apparent in many physical systems that normally occur in nature. 
This hypothesis is substantiated by the empirical success of natural cooling, such as using refrigerators, in reaching thermal equilibrium.

\vspace{1em}

\noindent \textbf{Quantum Gibbs sampling.} Aiming to mimic nature's cooling processes, a series of recent works have introduced quantum Markov Chain Monte Carlo (MCMC) algorithms, or quantum Gibbs samplers \cite{chen2023efficient, chen2023QThermalStatePrep, shtanko2021preparing, wocjan2023szegedy, Rall2023thermalstate, jianggibbssampling, zhang2023Gibbs, ding2024efficient, gilyen2024quantumMet}, as promising alternatives for tackling a range of classically intractable low-temperature simulation tasks on quantum computers.
These algorithms are designed to replicate the success of classical Markov chains in preparing Gibbs states for classical Hamiltonians. 
The analysis of classical MCMC algorithms relies on the principle of detailed balance; however, achieving this in the quantum setting has been challenging and was only recently addressed by an algorithm in \cite{chen2023efficient}.
Part of the difficulty arises from a conflict between the finite energy resolution $\sigma_E$ achievable by efficient quantum algorithms and the seemingly strict requirement to precisely distinguish energy levels to satisfy detailed balance.
In this work, we focus primarily on this algorithm, referring to it as the CKG algorithm or the quantum Gibbs sampler when the context is clear.
We give a detailed review of this algorithm in Section~\ref{sec:CKGdetails} and Appendix~\ref{sec:CKGdetails2}.

The Gibbs sampling algorithm provides a fully general method for preparing Gibbs states by evolving an initial state $\vrho_0$ under a Lindbladian $\CL_{\beta}$, which is efficiently implementable on a quantum computer and produces the state \[\vrho_t = e^{\CL_{\beta} t}[\vrho_0]\] after time $t$. 
The runtime of the quantum Gibbs sampler is governed by the \emph{mixing time} of the corresponding quantum Markov chain, which is roughly the time required for $\vrho_t$ to approach the Gibbs state $\vrho_{\beta}$.
This in turn is bounded by the spectral gap $\lambda_{\text{gap}}(\CL_\beta)$ of the Lindbladian by 
\[t_{\text{mix}}(\CL_\beta)\le \frac{\mathcal{O}(\beta \norm{\vH}+ \log(n) )}{\lambda_{\text{gap}}(\CL_{\beta})}.\]

The mixing time varies based on the quantum system in question. Bounding this mixing time is challenging without access to fault-tolerant quantum computers, as we cannot run and benchmark the algorithm directly, making theoretical analysis essential. However, such analysis is hindered by a lack of technical tools for two key reasons.
Firstly, the theory of convergence of quantum Markov chains is new, unlike the very mature twin field for classical Markov chains.
Secondly, the Markov chain described by the algorithm is considerably complex, and depends on several parameters that we will discuss in more detail shortly: an energy resolution $\sigma_E$, a series of jump operators $\vA^a$ for $a\in [M]$, and the inverse temperature $\beta$. 
The space of possibilities makes the algorithm's performance more difficult to characterize. 

This motivates the identification of quantum systems whose mixing times are tractable for analysis yet exhibit rich features that provide insights into the performance of the quantum Gibbs sampler for more general non-commuting Hamiltonians.
In line with this, the mixing time of the CKG algorithm has recently been bounded for local Hamiltonians, showing a polynomial scaling with system size at high enough temperatures~\cite{rouze2024efficient}.

\vspace{1em}

\noindent \textbf{Mixing time of sparse Hamiltonians.} In this work, we consider an alternative approach by characterizing the mixing time of a family of \emph{sparse} Hamiltonians of the following form:
\begin{align}
    \vH = \sum_{i, j \in [n]} H_{ij}\ket{e_i} \bra{e_j},\label{eq:sparseH}
\end{align}
which can be understood as the Hamiltonian on a graph $G = (V,E)$ with $n=|V|$ vertices indexed by basis states $\ket{e_i}$, $i\in [n]$ and a set of edges $E$ connecting vertices with $H_{ij} \neq 0$.
When non-zero entries $H_{ij}$ are all equal to $1$, the Hamiltonian $\vH$ corresponds to the $n\times n$ adjacency matrix of the $n$-vertex graph. 
We define the degree $d$ of the graph $G$ as the sparsity of the underlying Hamiltonian and refer to Hamiltonians with constant or slowly increasing degrees $d=\polylog(n)$ as \emph{sparse}. 
Note that any $\log(n)$-qubit Hamiltonian that consists of $m = \polylog(n)$ terms each acting locally on $\kappa = O(1)$ qubits is a sparse Hamiltonian with degree~$d \leq m 2^\kappa\leq \polylog(n)$. However, not all sparse Hamiltonians admit local qubit encodings.

Having defined sparse Hamiltonians, we now consider the dissipative dynamics of this system induced by a set of $M$ jump operators expressed as follows:
\begin{align}
    \vA^a =  \sum_{i,j \in [n]}A_{ij}^a \ket{e_i}\bra{e_j}, \quad \quad a \in [M].
\end{align}

We will soon explain how the jump operators $\vA^a$ relate to the Lindbladian $\CL_\beta$.
Briefly, the resulting dynamics can be understood as a combination of two processes: a continuous-time quantum walk of a single particle on the graph of states due to the coherent evolution of the Hamiltonian $\vH$, which is combined with stochastic jumps on the graph determined by the jump operators $\vA^a$. 

Our interest in bounding the mixing time of the sparse Hamiltonians is multifaceted: 
\begin{itemize}
    \item[(1)] \textbf{Single-particle dynamics.} As stated earlier, bounding the mixing time of general interacting multipartite Hamiltonians is a challenging task.
    However, for simple choices of graphs $G$, the mixing time of the quantum Gibbs sampler may be easier to analyze, potentially leading to relevant techniques for tackling the case of interacting particles.
    In fact, we can think of the dynamics induced by the Hamiltonian $\vH$ \eqref{eq:sparseH} as the dynamics of a single-particle hopping on the graph $G$.
    This single-particle evolution on path graphs or grids is commonly analyzed in the tight-binding model in condensed matter physics. 
    That being said, even in the simplified case of a single particle, the Hamiltonian $\vH$ is non-commuting, characterizing a continuous-time quantum walk that can yield exponential quantum advantage for certain oracular problem on graphs such as the glued tress \cite{Childs2003Glued}.
    \item[(2)] \textbf{Chaotic Hamiltonians.}
    Our additional motivation for studying random sparse Hamiltonians stems from the fact that their spectra exhibit many of the same characteristics as \emph{chaotic} Hamiltonians, such as the SYK model \cite{SYK, kitaev2015SYK, kitaev2015simple} and random $p$-spin models \cite{swingle2024bosonic, winer2022spectral}. 
    Understanding whether chaotic Hamiltonians have a fast mixing time as they approach their thermal and low-energy states is a fundamental question in the study of quantum chaos \cite{ETH_thermalization_Chen21, anshuetz2024strongly}. 
    As a concrete step toward addressing this problem, we identify key spectral properties of random sparse Hamiltonians that can ensure a fast mixing time.
    \item[(3)] \textbf{Algorithmic applications.}
    Preparing quantum Gibbs states, and more broadly computing the matrix exponential of sparse matrices such as the adjacency or Laplacian of a graph, is a fundamental subroutine in solving various graph and optimization problems. 
    For instance, the Estrada index---defined as the trace of the matrix exponential of a graph's adjacency matrix---measures subgraph centrality and provides structural insights \cite{Estrada2005centrality}.
    Computing the matrix exponential is also related to matrix inversion and linear system solvers \cite{sachdeva2013matrix}. 
    Moreover, quantum Gibbs sampling has been applied to solving semidefinite programs (SDPs) in optimization problems \cite{GSLBrandao2022fasterquantum, brandao2016QSDPSpeedup, brandao2017QSDPSpeedupsLearning, apeldoorn2017QSDPSolvers}, offering quantum speedups for these problems.
\end{itemize}

\section{Our main results} 

Motivated by these considerations, we investigate the mixing time of quantum Gibbs samplers for sparse Hamiltonians and different choices of jump operators. Our study addresses two key questions regarding the performance of the quantum Gibbs sampler for sparse Hamiltonians.
First, we ask

\begin{center}
    \emph{What choices of jump operators leads to a fast mixing time?}
\end{center}
After exploring the effects of different jump operators $\vA^a$, we then focus on the spectral properties of sparse Hamiltonians to understand:
\begin{center}
    \emph{What spectral property of the Hamiltonian determines its mixing time?}
\end{center}
Answering these questions allows us to provide a broad and intuitive insights of how the quantum Gibbs sampler operates for general families of sparse Hamiltonians.

\subsection{Choice of jumps: graph-local vs unitary design} 
A natural set of jump operators for a given $n \times n$ Hamiltonian on a graph $G$ are $\vA^a = \frac{1}{\sqrt{n}}\ketbra{e_a}{e_a}$ or similar operators supported on a few neighboring vertices of $G$.
Importantly, these are not ``local'' in the sense of multi-particle Hamiltonians, which refers to being composed of terms that act on a small number of qubits---often also geometrically close to one another.
Utilizing graph-local jump operators also significantly simplifies the structure of the Lindbladian and the analysis of mixing times for certain graph families.

Moving beyond graph-local jumps, the Lindbladian $\CL_\beta$ of the quantum Gibbs sampler can still be efficiently implemented on a quantum computer with certain non-local jumps.
This is possible as long as each jump $\vA^a$ is efficiently implementable, the set of jumps $M$ includes both $\vA^a$ and its adjoint $\vA^{a \dag}$, and $\sum_{a \in [M]} \norm{\vA^{a\dag}\vA^a}_{\infty}=1$ (due to this normalization condition, we will sometimes speak of the jumping distribution $\mathcal{A}$, from which the jump operators $\vA^a$ are sampled with probability $\norm{\vA^{a\dag}\vA^a}$). 
This raises the question of whether there is an advantage in using \emph{non-local} jumps that have a bounded spectral norm, or if more structured local jumps are sufficient to achieve a fast-mixing quantum MCMC. After all, \emph{classical} continuous-time random walks are typically considered with local jumps on the graph vertices.
However, in the context of graphs, we will see that the structured nature of local jumps offers no advantage, but rather seems to cause a slowdown of the resulting algorithm.

\vspace{1em}

\noindent\textbf{Graph-local jumps.} To this end, in the next theorem, we establish tight bounds on the spectral gap of the Lindbladian for cyclic graphs for graph-local jumps $\vA^a$, with an approach similar to the one used in \cite{temmedaviesgap2013} to bound the spectral gap of a Davies generator.

\begin{thm}[\textbf{Spectral gap of cyclic graphs with local jumps}]\label{thm:cyclicgapn_intro}
    Fix temperature $\beta^{-1}$. There exists some constant energy resolution $\sigma_E$ for which the spectral gap of the CKG Lindbladian $\mathcal{L}_\beta$ for a cyclic graph with $n$ vertices with jump operators $\vA^a = \frac{1}{\sqrt{n}}\ketbra{e_a}{e_a}$ is asymptotically $\Theta(n^{-3})$.
\end{thm}

In addition to theoretical analysis, we also generated data for cyclic graphs, path graphs, and random $d$-regular graphs with $n$ vertices, as shown in Figure \ref{fig:ckg_graph_numerical}. 
These numerical results suggest spectral gaps of $o(n^{-1})$ for generic sparse graphs with graph-local jumps. 
We observed that increasing the constant $d$ does improve the spectral gap decay, though it never improved past the asymptotic decay $O(n^{-1})$. 

These results are all suboptimal, since for an $n\times n$ Hamiltonian $\vH$, we expect an efficient result would be polynomial in the number of qubits, i.e. $\polylog(n)$ rather than $\poly(n)$. 
The poor performance can be attributed to two factors: (1) The operators $\vA^{a\dag} \vA^a$ have $L^1$ norm $\frac{1}{n}$. (2) In the energy basis, many entries of $\vA^a$ are highly correlated.

The first drawback effectively scales the Lindbladian down by $\frac{1}{n}$, since the $L^1$ norm of $\vA^\dag \vA$ can be as high as $1$ when the operator norm $\norm{\vA^\dag \vA} = \frac{1}{n}$. 
However, the chosen jump operators are projectors, so their $L^1$ and operator norms are equal.  
In both Theorem~\ref{thm:cyclicgapn_intro} and in the data, a spectral gap even worse than $\frac{1}{n}$ is observed. 
This is due to the second drawback. The aforementioned correlations lead to off-diagonal terms in the Lindbladian, which in general have the potential to dampen the spectral gap, and in the case of the cyclic graph provably do so. 
It appears that more generally, the biases of an ensemble of local jumps can introduce off-diagonal terms to the Lindbladian that decrease the spectral gap. 
The same harmful correlations appears to exist in higher degree graphs in addition to cyclic ones, though to a lesser extent as evidenced by the improved spectral gap.

\begin{figure}[tb!]
    \centering
    \textbf{CKG Lindbladian Spectral Gap Data}\par\medskip
    \includegraphics[width=1\linewidth]{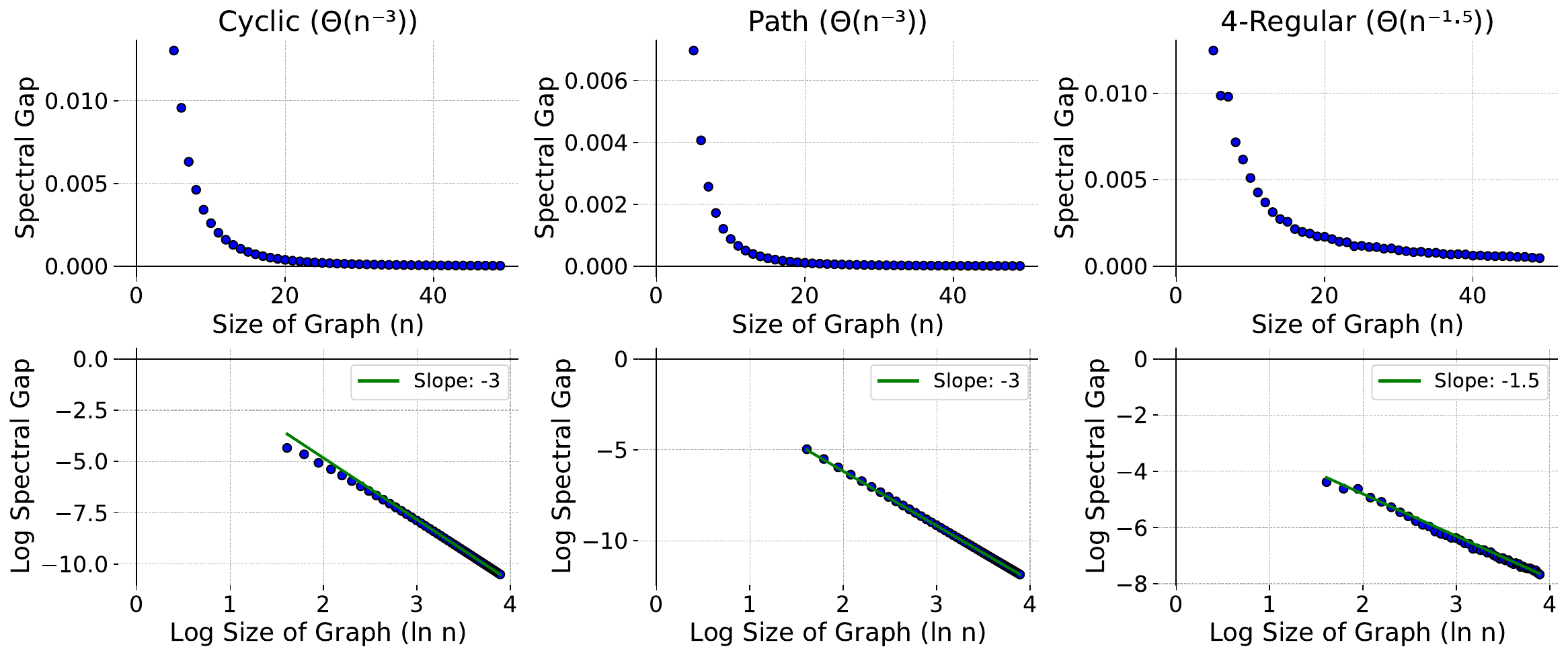}
    \caption{Linear (above) and log-log (below) graphs of spectral gap with respect to system size. Gaps of ten random 4-regular graphs were averaged for each data point. For the cyclic graphs (one-dimensional lattices), the proven asymptotic decay aligns closely with the data. }
    \label{fig:ckg_graph_numerical}
\end{figure}

\vspace{1em}
\noindent\textbf{Unitary design jumps.} To address some of these shortcomings, we next consider non-local jump operators, each independently drawn (along with its adjoint pair) according to a unitary $1$-design $\mathcal{D}(U(n))$ on $n$ vertices. 
More precisely, we define
\begin{defn}\label{def:design}
    A set of jump operators $\{\vA^a: a \in [M]\}$ is \textbf{drawn from a }$\mathbf{1}$\textbf{-design jumping distribution} if it is obtained by sampling $M/2$ jump operators i.i.d from a unitary $1$-design $\mathcal{D}(U(n))$, normalizing each by $\frac{1}{\sqrt{M}}$, and including these operators along with their adjoint. 
\end{defn}

We include the adjoint of each randomly chosen jump since the CKG Lindbladian requires the set of jump operators to be closed under adjoint, $\{\vA^a: a\in [M]\} = \{\vA^{a\dag}: a\in [M]\}$. 
When $n$ is a power of $2$, the unitary $1$-design can be constructed as a tensor product of random Pauli operators on $ \log_2(n)$ qubits, in which case the jumps are self-adjoint and can be sampled and implemented efficiently. 
The efficiency of our results on a general system relies on the ability to efficiently implement some unitary 1-design. 

As we will see, in our application, this 1-design sampling is effectively equivalent to sampling from a Haar-random distribution.
This approach improves on the results given for graph-local jumps, and is able to achieve an efficient algorithm in the number of qubits for a graph (running time $\polylog(n)$) for Gibbs sampling.
This improved performance is in part because all the eigenvalues of a Haar random unitary have magnitude 1. 
Hence, it avoids the problem of $\vA^\dag \vA$ having a relatively small $L^1$ norm given the constraint on its operator norm $\norm{\vA^\dag \vA}$.
These jumps also avoid the second problem encountered for the graph-local jumps: Since the number of degrees of freedom of randomness is very large over the ensemble, any form of bias is mitigated. 
Indeed, the resulting Lindbladian over the full ensemble has no off-diagonal terms resulting from correlated elements of the jump operator.

Our results extend beyond cyclic graphs to any graph of bounded degree $d = O(1)$ where $\norm{\vH}\leq d$ at constant temperature, or more generally when $\beta\norm{\vH}= O(1)$. 
We refer to these sparse Hamiltonians as bounded degree and formally define them as:

\begin{defn}
A \textbf{bounded degree} system is a sequence of temperatures $\beta(n)^{-1}$ and Hamiltonians $\vH(n)$ for which $\beta(n)\norm{\vH(n)}$ is bounded from above by a constant independent of system size. 
\end{defn}
\begin{thm}[\textbf{Constant spectral gap of Lindbladian in bounded degree systems}]\label{thm:constgapbounded_intro}
    Let $\beta(n)^{-1}$ be a sequence of temperatures and $\vH(n)$ a sequence of $n \times n$ Hamiltonians such that $\beta(n)\norm{\vH(n)} = O(1)$.
    
    With any constant probability $1-\xi$, the spectral gap of a Lindbladian $\mathcal{L}_\beta$ with $\sigma_E = \beta^{-1}$ and $M$ jump operators sampled from a 1-design jumping distribution for some $M = \Theta(\log(n))$, is bounded below by a constant, i.e. $\lambda_\text{gap} = \Omega(1)$.

     As a consequence and in the same setup,  the Gibbs state of $\vH$ can be prepared with error $\epsilon$ in diamond distance, in time $\poly(\log(n), \log(\epsilon^{-1}))$. 
\end{thm}

While the examples of bounded degree Hamiltonians we consider are mostly graphs, the above theorem applies to preparing the Gibbs state for any Hamiltonian at a temperature $\beta^{-1}$ such that $\beta^{-1}=O(\norm{\vH})$.

\subsection{Mixing time from spectral profile}
Theorem~\ref{thm:constgapbounded_intro} demonstrates that bounded-degree Hamiltonians with non-local jumps exhibit fast mixing times.
However, it leaves open the case of Hamiltonians with unbounded degrees, such as those with $d \leq \polylog(n)$.
More formally, we define

\begin{defn}
An \textbf{unbounded degree} system is a sequence of temperatures $\beta(n)^{-1}$ and Hamiltonians $\vH(n)$ for which $\lim_{n \to \infty}\beta(n)\norm{\vH(n)} \to \infty$. However, we still assume that $\beta\norm{\vH} = \polylog(n)$, polynomial in the number of qubits. 
\end{defn}

In our next result, we show that again selecting the jumping distribution (including adjoints) to be $M$ samples from a unitary 1-design for sufficiently large $M$ and choosing the energy resolution $\sigma_E = \beta^{-1}$ yields an algorithm whose efficiency depends on its low-energy spectrum. In particular, the runtime scales inverse polynomially with the fraction of eigenvalues $\delta(n)$ of $\vH$ that are within $O(\beta^{-1})$ of the minimum eigenvalue. 

\begin{thm}[\textbf{Spectral gap of Lindbladian in unbounded degree systems}]\label{thm:gapunbounded_intro}
    Let $\beta(n)^{-1}$ be a sequence of temperatures and $\vH(n)$ a sequence of $n \times n$ Hamiltonians, and let $\delta(n)$ be the fraction of eigenvalues of $\vH(n)$ within $O(\beta^{-1})$ of $\lambda_\text{min}$.
    
    With any constant probability $1-\xi$, the spectral gap of a Lindbladian $\mathcal{L}_\beta$ with $\sigma_E = \beta^{-1}$ and $M$ jump operators sampled from a 1-design jumping distribution for some $M = \Theta(\delta(n)^{-2}\log(n)\log(\beta\norm{\vH}))$, is lower bounded by $\Omega(\delta(n))$.

     As a consequence, and in the same setup, the Gibbs state of $\vH$ can be prepared with error $\epsilon$ in diamond distance, in time $\poly(\delta(n)^{-1}, \log(n), \log(\epsilon^{-1}))$.  
\end{thm}
\noindent
Note that if $\beta \norm{\vH}$ is bounded, then $\delta(n) = 1$. 
This result therefore generalizes Theorem~\ref{thm:constgapbounded_intro}. 

\subsection{Explicit examples of random sparse Hamiltonians} Having established a sufficient spectral condition for the fast mixing of random ensembles of sparse Hamiltonians, we now give explicit examples that satisfy this criterion. We also give one example, the hypercube, which does not, and for which local jumps in place of unitary design jumps achieve an exponential speedup. This example elucidates the potential of structured local jumps for speedups, in contrast to the case of the cyclic graph in which structured graph-local jumps yielded a slowdown.

\vspace{1em}

\noindent\textbf{Random regular graphs.} The first example is when $\vH$ is the adjacency matrix of a randomly selected $\log(n)$-regular graph, with $\polylog(n)$ random 1-design jumps.
In Section~\ref{sec:randomRegular}, we prove using Theorem~\ref{thm:gapunbounded_intro} that this ensemble has a Lindbladian spectral gap of $\Omega(\log(n)^{-3/4})$ at constant temperature. 
This yields a polynomial algorithm to prepare the Gibbs state.

\vspace{1em}

\noindent\textbf{Random signed Pauli ensemble.} The second example is the family of sparse Hamiltonians considered in \cite{chen2023sparse}, composed of random Pauli strings with random sign coefficients given by
\[\vH_{PS} = \sum_{j=1}^m \frac{r_j}{\sqrt{m}} \vsigma_j\]
where $\vsigma$ is a random Pauli string on $n_0$ qubits (such that the size of Hamiltonian is $n \times n$ for $n=2^{n_0}$), each $r_j$ is sampled randomly from $\{-1, 1\}$, and $m = O(\frac{n_0^5}{\epsilon^4})$ for a parameter $1\geq \epsilon \geq 2^{-o(n_0)}$.  
We show in Section~\ref{sec:randomPauli} that the CKG Lindbladian has a spectral gap of $\Omega(\epsilon^{-3/2})$ when we choose $M = \widetilde{O}\left(n_0^2\epsilon^{-3}\right)$ unitary 1-design jumps, inverse temperature $\beta = O(\epsilon^{-1})$, and $\sigma_E = \beta^{-1}$.

\vspace{1em}

\noindent\textbf{Hypercubes.} The final example is the family of hypercubes. 
A hypercube with $2^d$ vertices and degree $d$ can be interpreted as a Hamiltonian on $d$ qubits $\sum_i \vX_i$. At constant temperature, only an exponentially small fraction of eigenvalues lie near the minimum eigenvalue. 
As a result, the spectral profile implies a poor mixing time with unitary design jumps. 

However, we show in Theorem \ref{thm:hypercubelocalgap} that by choosing local jumps $\frac{1}{\sqrt{d}}\vZ_i$, the spectral gap is $\frac{1}{d}$, yielding an efficient algorithm for Gibbs sampling. 
Hypercubes therefore provide an example where not graph-local jumps, but rather local jumps on the qubits, ensure fast mixing. The crucial feature of local jumps that improves mixing time is that a local jump $\vA^a$ on a local Hamiltonian $\vH$ satisfies $\norm{[\vA^a, \vH]} = O(1)$—i.e., the jump operator only jumps between nearby eigenstates. 
This property is not held by graph-local jumps in general, so they displayed no improvement in the studied cases. In general, local jumps are the strongest candidates for fast-mixing on local Hamiltonians, though for which classes of local Hamiltonians fast-mixing can be achieved is still largely open.
\section{Proof sketch}

\subsection{Graph-local jumps}
The proof of the mixing time for graph-local jumps in the cyclic graph involves two steps:
First, in Appendix~\ref{sec:CKGinEigenbasis}, we derive a general expression for the CKG Lindbladian in the energy basis given by equation~\eqref{eq:graphlindbladian}.

We then utilize this expression along with the fully known spectrum of the cyclic graph to show that the Lindbladian is block-diagonal in the energy basis of this graph, as demonstrated in Appendix~\ref{cyclicgapnproof}. One of these blocks corresponds to a classical Markov chain on the diagonal entries of the state in the energy basis, for which we establish a spectral gap lower bound using the canonical path method. For the remaining $n-1$ blocks, we apply the Gershgorin circle theorem to bound their eigenvalues.

\subsection{Unitary design jumps}

\noindent\textbf{Bounded-degree systems.} To establish a lower bound on the spectral gap of the Lindbladian $\CL_\beta$ with unitary $1$-design jumps, we consider a decomposition of the form: 
\begin{align}
\CL_\beta = \CL_\mu + \delta\CL, \label{eq:decomposeL}
\end{align}
where $\CL_\mu = \mathbb{E}_{\vA \sim \mathcal{D}(U(n))}[\mathcal{L}_\beta]$ is the expected Lindbladian with the expectation taken over a single jump operator sampled from a unitary $1$-design distribution $\mathcal{D}(U(n))$, and $\delta\CL$ represents the remainder term.
Due the quadratic form of the Lindbladian given in expression~\eqref{eq:lindbladianForm} we see that $\mathbb{E}_{\vA \sim \mathcal{H}(U(n))}[\mathcal{L}_\beta] = \mathbb{E}_{\vA \sim \mathcal{D}(U(n))}[\mathcal{L}_\beta]$.
Here, $\mathcal{H}(U(n))$ is a Haar random distribution over jump operators. 

Note that a CKG Lindbladian must have jump operators in adjoint pairs, so a Lindbladian with a single jump operator will not satisfy detailed balance. 
However, the expected Lindbladian over one Haar random jump operator is equal to the expected Lindbladian over the adjoint pair of a Haar random jump operator, by linearity of expectation. 

The proof of Theorem~\ref{thm:constgapbounded_intro} proceeds by first showing that this expected Lindbladian $\mathbb{E}_{\vA \sim \mathcal{H}(U(n))}[\mathcal{L}_\beta]$ has a constant spectral gap, as long as $\beta(n)\norm{\vH(n)}$ is bounded by a constant as a function of system size. 
As before, we call such systems \textit{bounded degree}, since for constant $\beta$ bounded degree graphs satisfy the required property. 
Indeed, if such a system has degree bounded by $d$, it must have spectrum in $[-d, d]$ by Gershgorin's circle theorem, since every row consists of zeros and at most $d$ ones. 
Adding phases to the edges of these bounded degree graphs remains feasible by a similar argument, so there is no constraint of stoquasticity. 

The result is stated formally as follows:
\begin{lem}[\textbf{Constant spectral gap of average Lindbladian for bounded degree systems}]\label{lem:boundedaveragelindbladian}
    Let $\beta(n)^{-1}$ be a sequence of temperatures and $\vH(n)$ be a sequence of $n \times n$ Hamiltonians such that $\beta\norm{\vH} = O(1)$. 
    The spectral gap of $\mathcal{L}_\mu = \mathbb{E}_{\vA \sim \mathcal{D}(U(n))}\left[\mathcal{L}_\beta\right]$, the expected CKG Lindbladian with energy resolution $\sigma_E= \beta^{-1}$ over the ensemble of one jump operator sampled from a unitary 1-design, is asymptotically $\Omega(1)$.
\end{lem}

To establish Lemma~\ref{lem:boundedaveragelindbladian}, we show that for any system, the average Lindbladian over a Haar random ensemble of jump operators decomposes as $\mathcal{L}_\mu= \mathcal{L}_\text{classical} + \mathcal{L}_\text{dephasing}$. 
For a density matrix in the energy basis, the evolution of $\mathcal{L}_\text{classical}$ is a classical continuous Markov chain of the diagonal.
The evolution of this classical Lindbladian maps the diagonal, in the limit, to the Gibbs distribution. 
The spectral gap of $\mathcal{L}_\text{classical}$ can be analyzed with the large suite of techniques for classical Markov chains. 

Meanwhile, $\mathcal{L}_\text{dephasing}$ damps the off-diagonal terms of the density matrix. 
In the limit as $t \to \infty$, the state therefore converges to a classical distribution on the diagonal in the energy basis, with no off-diagonal terms, as desired. 
The operator $\mathcal{L}_\text{dephasing}$ diagonalizes in the energy basis of density matrices, with each off-diagonal element decaying at an independent rate. 
It is therefore simple to analyze as well.
In summary, Lemma~\ref{lem:boundedaveragelindbladian} establishes that the $\lambda_\text{gap}(\mathcal{L}_\mu) = \min(\lambda_\text{gap}(\mathcal{L}_\text{classical}), \lambda_\text{min}(\mathcal{L}_\text{dephasing})) = \Omega(1)$.

However, this result does not imply that any given jump sampled from the unitary 1-design (with its adjoint pair) would yield a gapped Lindbladian. 
As a result, it does not yet yield an efficient Gibbs sampling algorithm. 
To obtain such a result in Theorem~\ref{thm:constgapbounded_intro}, we demonstrate that the remainder term $\delta \CL$ in \eqref{eq:decomposeL} has a small spectral norm when a Lindbladian is constructed from a sufficiently large number of jumps $M$, rather than just one.
In particular, a Lindbladian sampled with $\Theta(\log(n))$ normalized jumps from any 1-design concentrates closely to its expectation, thereby establishing a spectral gap lower bound.
Since this lower bound applies to any graph at constant temperature $\beta^{-1}$ with bounded degree, it applies to the periodic lattices, path graphs, and $k$-regular graphs discussed in the previous section.

\vspace{1em}

\noindent \textbf{Unbounded degree systems.} In the context of unbounded degree systems, 1-design unitaries can no longer, in general, achieve an algorithm that is efficient in $\log(n)$. 
Indeed, $\mathcal{L}_\mu = \mathbb{E}_{\vA \sim \mathcal{D}(U(n))}\left[\mathcal{L}_\beta\right] = \mathcal{L}_\text{classical} + \mathcal{L}_\text{dephasing}$ does not necessarily have a constant spectral gap in general, as it did in the case of bounded degree systems. 
However, we may establish a condition on the spectrum of $\vH$, with which we can recover a lower bound for the spectral gap:

\begin{lem}[\textbf{Spectral gap of average Lindbladian for unbounded systems}]\label{lem:averagelindbladian}
    Let $\vH(n)$ be a sequence of $n \times n$ Hamiltonians. For some $C$, let $\delta(n)$ be the proportion of eigenvalues $\lambda_j$ of $\vH$ such that $\beta^{-1} (\lambda_j-\lambda_\text{min}) \leq C$. 
    The spectral gap of $\mathcal{L}_\mu = \mathbb{E}_{\vA \sim \mathcal{D}(U(n))}\left[\mathcal{L}_\beta\right]$, the expected CKG Lindbladian over the Haar random unitary ensemble of its jump operator at temperature $\beta^{-1}$ with $\sigma_E = \Theta(\beta^{-1})$, is asymptotically $\Omega\left(\delta(n)\right)$.
\end{lem}

The lemma expresses that if $\lambda_\text{min}$ is within $O(\beta^{-1})$ of $\delta(n)$ of the eigenvalues, the spectral gap is at least $\delta(n)$.
Similarly to Theorem~\ref{thm:constgapbounded_intro}, using this result to obtain an efficient Gibbs sampling algorithm amounts to showing that a Lindbladian with enough independently sampled jump operators shares a similar asymptotic spectral gap to the average Lindbladian, using a concentration bound. 
When the average spectral gap from the above lemma is $\delta(n)$, the number of jump operators to concentrate around the expectation increases to $\delta(n)^{-2}$, along with an overhead of $\log(n)\log(\beta \norm{\vH})^2$. This result is captured in Theorem \ref{thm:gapunbounded_intro}, and results in an algorithm with runtime $\poly(\delta(n)^{-1}, \log(n),\log(\epsilon^{-1}))$ for Gibbs sampling, where $\epsilon$ is the error in diamond distance. This runtime bound relies on the standing assumption in this paper that $\log(\beta \norm{\vH}) = \poly(\log(n))$. 

\section{Technical details}
\subsection{The quantum Gibbs sampler}
\subsubsection{Lindbladian evolution}
The recently proposed CKG quantum MCMC algorithm addresses the problem of finding thermal states by imitating natural thermodynamic processes \cite{chen2023efficient, chen2023QThermalStatePrep}. 
In this process, a system of particles evolve in contact with a thermal bath at some fixed temperature $\beta^{-1}$.
Due to the interactions with the bath, the state of the system is described by a probabilistic mixture of quantum states $\vrho$.
This state evolves in time, by approximation, with Markovian \emph{dissipative} dynamics \begin{equation}
    \frac{\rd \vrho}{\rd t} = \CL_{\beta}[\vrho],\label{eq:dissipativeEvolution}
\end{equation}
given in terms of an operator $\CL_{\beta}$ known as the Lindbladian. 
This operator involves a coherent term $\vB$ that describes the interaction among the particles in the system. 
There is also a term in $\CL_{\beta}$ that is specified by a series of \textit{Lindblad operators} $\vL^j$ that drive the dissipative transitions. 
The dynamics of the coherent term $\vB$ are reversible, while the dissipative transitions drive all states toward some ``stationary state''. 
These transitions can be understood as perturbations from the bath, and as all states converge to the Gibbs state, the information of the system is leaking via these perturbations to the bath. 
The expression, in terms of $\vB$ and $\vL^j$, is:

\[\mathcal{L}_\beta[\cdot ] = -i[\vB, \cdot] + \sum_j \left(\vL^j (\cdot) \vL^{j \dag} - \frac{1}{2}\{\vL^{j \dag} \vL^j, \cdot\}\right).\]
The summands $\vL^j (\cdot) \vL^{j\dag}$ are termed the transition part of the Lindbladian, and $- \frac{1}{2}\{\vL^{j \dag} \vL^j, \cdot\}$ are the decay part of the Lindbladian. The choice of the Lindbladian operator $\CL_{\beta}$ can vary depending on the precise nature of interactions between the system and the bath. However, to prepare the Gibbs (thermal) state at temperature $\beta^{-1}$, the Lindbladian should be designed to satisfy 
\begin{align}
\frac{\rd \vrho_{\beta}}{\rd t} = \CL_{\beta}[\vrho_{\beta}] = 0 \quad\text{where}\quad \vrho_\beta := \e^{-\beta \vH}/\tr(\e^{-\beta \vH}),\label{eq:main_Gibbs_fixed_point}
\end{align}
and moreover $\vrho_\beta$ should be the unique stationary state of the Lindbladian. The long-term evolution of the system under this Lindbladian, as a result, would converge to the Gibbs state of the Hamiltonian $\vH$ at temperature $1/\beta$. 
\subsubsection{Detailed balance}
To ensure that the Lindbladian $\mathcal{L}_\beta$ converges to a state $\vrho_\beta$,  \cite{chen2023efficient} designs a Lindbladian that satisfies Kubo-Martin-Schwinger (KMS) detailed balance with respect to $\vrho_\beta$. 
KMS detailed balance is one of several ways of quantizing the notion of classical detailed balance for Markov chains.
KMS detailed balance of $\mathcal{L}_\beta$ is self-adjointness with respect to the inner product 
\begin{equation}
\langle \sigma_1, \sigma_2 \rangle_{\rho_\beta^{-1}} = \tr(\sigma_1^\dag \vrho_\beta^{-1/2}\sigma_2\vrho_\beta^{-1/2}) \tag{KMS Inner Product}.
\end{equation}
In particular, it is equivalent to the relation that 

\begin{equation}
\mathcal{L}_\beta[\cdot] = \vrho_\beta^{1/2}\mathcal{L}_\beta^\dag\left[\vrho_\beta^{-1/2} (\cdot) \vrho_\beta^{-1/2}\right]\vrho_\beta^{1/2}\tag{Detailed Balance}\label{detbalance}
\end{equation}
where $\mathcal{L}_\beta^\dag$ is the adjoint Lindbladian with respect to the Hilbert-Schmidt inner product $\langle \sigma_1, \sigma_2 \rangle = \tr(\sigma_1^\dag \sigma_2)$.
The adjoint operator $\mathcal{L}_\beta^\dag$, in the Heisenberg picture, describes the dynamics of observables under evolution by  $\mathcal{L}_\beta$.  
The Lindbladian evolution is described by some quantum channel and therefore the observable $I$ must always be fixed by $\exp(\mathcal{L}_\beta^\dag)$. 
This implies that $\mathcal{L}_\beta^\dag[I] = 0$. 
The detailed balance formula thereby implies that $\mathcal{L}_\beta[\vrho_\beta] = 0$, as desired. 
Note that KMS detailed balance can be dually described as the self-adjointness of $\mathcal{L}_\beta^\dag$ with respect to the inner product $\langle \sigma_1, \sigma_2 \rangle_{\rho_\beta} = \tr(\sigma_1^\dag \vrho_\beta^{1/2}\sigma_2\vrho_\beta^{1/2})$.

\subsubsection{Construction and parameters}\label{sec:CKGdetails}
The quantum Gibbs sampler in \cite{chen2023efficient} constructs a Lindbladian that satisfies the following two properties:

\begin{enumerate}
\item $\mathcal{L}_\beta$ satisfies detailed balance with respect to $\vrho_\beta$, and therefore $\mathcal{L}_\beta[\vrho_\beta] = 0$. 
\item The dynamics of $\mathcal{L}_\beta$ can be efficiently implemented. 
\end{enumerate}

Their Lindbladian, which we term the CKG Lindbladian, can be simulated on a quantum computer with a cost per unit time $t=1$ roughly equal to that of simulating the Hamiltonian dynamics of $\vH$. 
The CKG Lindbladian is closely related to the Davies generator, which is a physically motivated Lindbladian that satisfies detailed balance, but that is not efficiently implementable in general. 
A full description of both Lindbladians are given in the \hyperref[sec:appendix]{Appendix}. 

The Gibbs sampling algorithm evolves an initial state $\vrho_0$ according to the efficiently implemented Lindbladian $\CL_{\beta}$, and produces the state $$\vrho_t = e^{\CL_{\beta} t}[\vrho_0]$$ after time period $t$.
The \emph{mixing time} is roughly the time that it takes for the state $\vrho_t$ to approach the Gibbs state $\vrho_{\beta}$. That is, 
\begin{align}
    \e^{\CL_{\beta}t_{\text{mix}}}[\vrho_0] \approx \vrho_{\beta}.
\end{align}
The efficiency of the algorithm therefore scales linearly with the unit time simulation cost and the mixing time.
The algorithm has several parameters in the Lindbladian's construction. 
In addition to the \textit{inverse temperature} $\beta$, the algorithm specifies an \textit{energy resolution} $\sigma_E$. 
A salient feature of \cite{chen2023efficient}'s construction is that it can achieve detailed balance even though the algorithm only probes the energies of the Hamiltonian $\vH$ with approximate precision. 
$\sigma_E$ quantifies this level of precision.  
The cost of the Lindbladian simulation depends linearly on $\sigma_E^{-1}$, but increasing the precision may also improve the mixing time. 
Taking $\sigma_E \rightarrow 0$ for absolute precision recovers the Davies generator—when distinguishing the energies of the system exactly is infeasible, this Lindbladian cannot be simulated efficiently.

A set of \textit{jump operators} $\vA^a$ must also be specified for the Lindbladian. 
These operators are decomposed by frequency and reassembled in a particular way to construct the Lindblad operators that help $\mathcal{L}_\beta$ satisfy detailed balance. 
They must appear in adjoint pairs: i.e., if $\vA \in \{\vA^a\}$, then $\vA^\dag \in \{\vA^a\}$. 
The cost of simulation scales with the cost of implementing the oracle $\ket{a} \rightarrow \ket{a} \otimes \vA^a$. 
In particular, the jump operators must be normalized when implemented for the algorithm, satisfying $\sum_a \norm{\vA^{a\dag}\vA^a} \leq 1$. 
CKG Lindbladians are linear in their jump operators—if $\mathcal{L}_1$ has one jump operator $\vA^1$ and $\mathcal{L}_2$ has one jump operator $\vA^2$, then a Lindbladian $\mathcal{L}$ with jump operators $\vA^1$ and $\vA^2$ satisfies $\mathcal{L} = \mathcal{L}_1 + \mathcal{L}_2$.
If $\mathcal{L}_\beta$ was constructed from jumps $\vA^a$, then jump operators $\sqrt{s}\vA^a$ produce the Lindbladian $s \mathcal{L}_\beta$, scaling the mixing time by $s$. 
So we may therefore assume that $\sum_a \norm{\vA^{a\dag}\vA^a} = 1$ exactly, since renormalizing can only improve the spectral gap.
In its normalized form, the set of jump operators can be understood as a jumping distribution over $\vA^a$ which we will notate $a \sim \mathcal{A}$, where each is sampled with probability $\norm{\vA^{a\dag}\vA^a}$.

\subsection{Spectral gap}
Since the quantum MCMC algorithm was proposed recently, numerical and analytic characterizations of algorithm are limited. 
As for classical Markov chains, it has been shown that the mixing time of the algorithm can be characterized by the spectral gap $\lambda_{\text{gap}}(\CL_{\beta})$ of the Lindbladian. 
If the first eigenvalue $\lambda_1=0$ corresponds to eigenvector $\rho_\beta$, then the spectral gap is $\lambda_{\text{gap}}(\CL_{\beta}) = \min_{j>1} |\lambda_j|$ \cite{chen2023efficient}.
Lindbladians are in general negative semidefinite like classical Markov chain generators, so $\lambda_\text{gap}(\CL_{\beta})= \min_j (-\lambda_\text{j})$.
More precisely, it holds that
    \begin{align}
        \frac{\Omega(1)}{\lambda_{\text{gap}}(\CL_\beta)} \le t_{\text{mix}}(\CL)\le \frac{\log\left(\left\Vert\rho_\beta^{-1/2}\right\Vert\right)}{\lambda_{\text{gap}}(\CL_\beta)} \le \frac{\mathcal{O}(\beta \norm{\vH}+\log(\text{dim}(\vH)))}{\lambda_{\text{gap}}(\CL_\beta)}.
    \end{align}

In particular, analytically bounding this spectral gap from below is sufficient for obtaining an upper bound on the~mixing~time.
For so-called rapid mixing, in which the mixing time is logarithmic in the number of qubits, the spectral gap bound often does not suffice. 
For our purposes of proving efficiency in the number of qubits, however, this issue is moot.

\subsection{Unbounded degree systems}

In this section, we further explore and prove efficient Gibbs sampling results for certain classes of unbounded degree sparse Hamiltonians.

\subsubsection{Random \textbf{log}$\mathbf{(n)}$-regular graphs}\label{sec:randomRegular}
With high probability at constant temperature, a randomly selected $d=\log(n)$-regular graph, with $\poly(d)$ random 1-design jumps, has a Lindbladian spectral gap of $\Omega(d^{-3/4})$. 
This gives a polynomial algorithm to prepare the Gibbs state for most such graphs at constant temperature.

The gap of $\Omega(d^{-3/4})$ arises because a random $d$-regular graph, for $d \to \infty$, has one eigenvalue at $d$ and the rest distributed from $-2\sqrt{d-1}$ to $2\sqrt{d-1}$ in a distribution that converges to a (normalized) semicircle.
This semicircular distribution frequently appears in random matrix theory, for instance in the Gaussian unitary ensemble (GUE), which models the spectrum of many chaotic quantum systems. 
When the spectrum of a quantum system indeed follows this distribution, it implies that $\delta(n) = \Omega(d^{-3/4})$ of the eigenvalues lie within a constant of the minimum eigenvalue. 

\begin{thm}\label{thm:regspecstat}
With any constant probability $1-\xi$, for a randomly selected $d=\log(n)$-regular graph, there are $\delta = \Omega(d^{-3/4})$ eigenvalues within $O(1)$ of the minimum eigenvalue. 
\end{thm}
As an immediate result of Theorem \ref{thm:regspecstat} and Theorem \ref{thm:gapunbounded_intro}, we obtain an algorithm polynomial in $d$ to prepare the Gibbs state of a $d$-regular graph. 
To prove the corollary, we use Theorem 2 in \cite{sparsereggraphs}. For this context, the following statement suffices. 
\begin{thm}\label{thm:sparseregthm}
Let degree $d = \log(n)$. For sufficiently large $n$, there exists some $D>0$ such that for any interval $I \subset \mathbb{R}$, $0<\alpha <1$, and $0<\epsilon<\alpha$ such that $|I|> Dd^{-\alpha+\epsilon}$, with probability $1-o(n^{-1})$ over all random $d$-regular graphs,
\[\left|\delta(n) - \mu \right| < d^{-\epsilon}|I|,\]
where $\mu = \int_I \rho_{sc}(x)dx$ and $\delta(n)$ is the fraction of eigenvalues of the $d$-regular graph in $I\sqrt{d-1}$. 
\end{thm}
In the above, $\rho_{sc}$ is the asymptotic distribution as $d \to \infty$ of a random $d$-regular graph is the semicircular distribution with radius $2\sqrt{d-1}$:
\[\rho_{sc}(x) = \frac{1}{2\pi (d-1)}\sqrt{4(d-1) - x^2}.\]

\noindent We may now prove Theorem \ref{thm:regspecstat}:

\begin{proof}
Writing $T = -2\sqrt{d-1}$, we calculate the mass of $\rho_{sc}$ from $T$ to $T+C$:
\begin{align*}
\mu = \int_{T}^{T+C} \rho_{sc}(x) dx &\leq \int_{T}^{T+C} \frac{\sqrt{4(d-1) - x^2}}{2\pi (d-1)}  dx\\
&\leq \int_{T}^{T+C} \frac{\sqrt{4(d-1) - x^2}}{2\pi (d-1)} dx\\
&= \Omega(d^{-3/4}).
\end{align*}
The $C$ above is some positive constant, for which we now try to bound the fraction of eigenvalues $\delta(n)$ within $C$ of the minimum eigenvalue. 
The final equality holds because the first $\epsilon$ fraction of a semicircle has mass $\sim \epsilon^{3/2}$, and in this case the circle has radius $2\sqrt{d-1}$, and we integrate a fraction of size $\frac{C}{2\sqrt{d-1}} = \Theta(d^{-1/2})$. 

We now apply Theorem \ref{thm:sparseregthm} for the interval $I = [-2, -2 + \frac{C}{2\sqrt{d-1}}]$ choosing $\alpha = \frac{11}{12}$ and $\epsilon = \frac{1}{3}$. 
Since the length of $I$ is $\frac{C}{2\sqrt{d-1}}$, we have that 
\[|I| > Dd^{-\alpha + \epsilon}\log(d) = Dd^{ - 7/12}\log(d)\] 
for some $D>0$ and sufficiently large $d$.
The theorem therefore applies that with probability $1 - o(n^{-1})$, implying that $|\delta(n) -\mu| < d^{-1/3}|I| = O(d^{-5/6}) = o(d^{-3/4})$.
Since $\mu = \Omega(d^{-3/4})$, we conclude that $\delta(n) = \Omega(d^{-3/4})$. 
\end{proof}
\subsubsection{Pauli String Ensemble}\label{sec:randomPauli}
We now mention another ensemble of Hamiltonians studied by \cite{chen2023sparse} in the context of low-energy state preparation. 
 
In \cite{chen2023sparse}, efficient low energy state preparation with phase estimation is demonstrated under the same conditions as our efficient Gibbs sampling in Theorem~\ref{thm:gapunbounded_intro}.
Indeed, if many eigenvectors are close to the ground state energy, as we require, then performing phase estimation on the maximally mixed state has a high probability of measuring a low-energy state, so low-energy state preparation is possible as well. 
They study the following ensemble of Hamiltonians on $n_0$ qubits:

\[\vH_{PS} = \sum_{j=1}^m \frac{r_j}{\sqrt{m}} \vsigma_j\]
where $\vsigma$ is a random Pauli string on $n_0$ qubits, each $r_j$ is sampled randomly from $\{-1, 1\}$, and $m = \floor{c_2\frac{n_0^5}{\epsilon^4}}$. The parameter $\epsilon$ satisfies $\epsilon \geq 2^{-n_0/c_1}$, and $c_1, c_2$ are absolute constants.
The resulting spectrum is again close enough to a semicircular distribution to obtain an efficient Gibbs sampler for certain temperatures that depend on $\epsilon$.
As $\epsilon$ decreases, Gibbs sampling becomes efficient for even larger values of $\beta$ (lower temperatures), since the ensemble's spectrum converges closer to a perfect semicircular distribution at the edge of the spectrum. 

Using the results in their paper, we establish that Gibbs sampling is efficient in $n_0$ for certain values of $\epsilon$ and corresponding temperatures $\beta^{-1}$. 
\begin{thm}
Say that $\vH(n_0)$ is sampled from the ensemble $\vH_{PS}$ on $n_0$ qubits, with $\epsilon = 2^{-o(n_0)}$ and $\epsilon \leq 1$.
With any constant probability $1-\xi$, for sufficiently large $n_0$, $\delta = \Omega(\epsilon^{3/2})$ fraction of the eigenvalues lie within $O(\epsilon)$ of the minimum eigenvalue. 
\end{thm}

\begin{proof}
We utilize two results from \cite{chen2023sparse}. 
Firstly, they argue that $\Pr[\norm{\vH(n_0)} \geq 2(1+\epsilon)] \leq \exp(-c_2 n_0)$ when $m \geq \frac{n_0^3}{\epsilon^4}$, which is satisfied in this case. 
With an arbitrary constant probability for sufficiently large $n_0$, therefore, $\norm{\vH(n_0)} \leq 4$, since $\epsilon \leq 1$. 
The second result is that with probability $1 - \exp(-c_3 n_0^{1/3})$, at least $\Omega(\epsilon^{3/2})$ of the eigenvalues satisfy $\lambda_i \leq (1-\epsilon)\lambda_\text{min}$ where $c_3$ is an absolute constant. 
With any large constant probability, we therefore have that
\[|\lambda_i - \lambda_\text{min}| \leq \epsilon \lambda_\text{min} \leq 4\epsilon = O(\epsilon)\]
for $\Omega(\epsilon^{-3/2})$ of the eigenvalues. 
\end{proof}

By Theorem $\ref{thm:gapunbounded_intro}$, we obtain a Gibbs sampling algorithm that is $\poly(\epsilon^{-1}, n_0)$ to prepare the Gibbs state at inverse temperature $\epsilon^{-1}$. 
We may rephrase this result in terms of $\beta$. 
For any polynomially large $\beta$, it provides a Pauli string ensemble of Hamiltonians, $H_{PS}$ with $\epsilon = \beta^{-1}$, for which with high likelihood preparing the Gibbs state is efficient in $n_0$, assuming access to a block-encoding of the Hamiltonian of interest. 

\subsubsection{Hypercube graphs}
For hypercube with varying dimension at a constant temperature, using unitary 1-design jumps would yield an exponentially large runtime. 
The spectrum of a hypercube with dimension $d$ and $2^d$ vertices consists of the integers $-d, -d+2\dots, d-2, d$. The eigenvalue $j$ has multiplicity $\binom{d}{\frac{d+j}{2}}$. 
In particular, for any constant $C$, only an exponentially small fraction of the eigenvalues $\delta(d)$ lie below $-d + C$. 
This leads to a naive algorithm with at worst exponential complexity in $d$. 

However, a better result can be obtained by considering the hypercube as a system of $d$ qubits. 
The graph with dimension $d$ has $2^d$ vertices, which can be considered length $d$ bitstrings.
With this representation, the adjacency matrix is then the sum of Pauli $X$ operators on each qubit, $\sum_{i=1}^d X_i$, since the hypercube has an edge between any two bitstrings of Hamming distance 1. 
Choosing $d$ jump operators as $\frac{1}{\sqrt{d}}Z_a$, the mixing time can therefore be improved to $\poly(d, \log(\epsilon^{-1}))$:

\begin{thm}[Spectral Gap for Hypercube with Local Jumps]\label{thm:hypercubelocalgap}
For fixed $\beta^{-1}$, there exists some energy resolution $\sigma_E$ such that the spectral gap of the CKG Lindbladian $\mathcal{L}_\beta$ for a $d$-dimensional hypercube with jump operators $\vA^a = \frac{1}{\sqrt{d}}Z_a$, is asymptotically $\Omega(d^{-1})$.
\end{thm}

In the case of the hypercube, the local jump operators $\frac{1}{\sqrt{d}}Z_i$ only jump between eigenstates whose eigenvalues differ by 1. 
This vastly improves the performance of the classical Markov chain and dephasing Markov chain within the Lindbladian.
However, the Lindbladian does not consist only of these two terms, as it did in the limit of independently sampled 1-design jumps. Off-diagonal terms do exist, and the presence of $Z_a$ for \textit{every} index is necessary to ensure that these off-diagonal terms do not completely eliminate the spectral gap. 
In some way, there must be ``enough uncorrelated'' local energy jumps to dampen these off-diagonal terms.
For more complicated local Hamiltonians, it is not clear what the corresponding constraint the jump operators might be to guarantee fast mixing. 
However, in general the choice of jump operators for local Hamiltonians seems to include a compromise between locality and suppression of correlations.

\section{Connection to previous work}
Our results show that for a Hamiltonian $\vH$ with temperature $\beta^{-1}$ such that some $\delta(n)$ fraction of the eigenstates are within $O(\beta^{-1})$ of the ground state energy, the CKG quantum Gibbs sampler with 1-design jumps efficiently prepares the Gibbs state with diamond distance error at most $\epsilon$. 
The running time scales polynomially with $\delta(n)^{-1}$, $\beta \norm{\vH}$, $\log(\epsilon^{-1})$, and the complexity of the block encoding of $\vH$.

This result is a baseline test that shows the CKG algorithm performs as well as other methods for preparing low-energy states of Hamiltonians. 
Indeed our spectral condition is precisely the same as a condition that ensures easy quantum phase estimation of a near-ground state. 
Namely, performing quantum phase estimation on the maximally mixed state can prepare a random eigenstate with its measured energy, and with probability $\delta(n)$ it must be within $O(\beta^{-1})$ of the minimum eigenvalue.
Obtaining $O(\delta(n)^{-1})$ samples and taking the minimum energy can therefore prepare a near ground-state eigenvector. This approach is the basis of the previous analysis of random sparse Hamiltonians in~\cite{chen2023sparse}.

Moreover, in \cite{chowdhury2016gibbssampling}, a quantum algorithm is presented that prepares the Gibbs state with a complexity that scales as $\poly(\frac{n}{\mathcal{Z}(\beta)}, \log(\epsilon^{-1}))$. If $\delta(n)$ of the eigenstates are within $O(\beta^{-1})$ of the ground state energy, then $\frac{n}{\mathcal{Z}(\beta)} = \Omega(\delta(n)^{-1})$, and therefore under such conditions, this algorithm efficiently prepares the Gibbs states as well.
Effectively, the CKG Gibbs sampler with ``generic'' 1-design jumps performs the same as previously developed algorithms—a potential advantage in cooling must arise from a smart (i.e., local and unbiased) choice of jump operators.

Finally, our conditions on the spectrum and the structure of random unitary design jumps resemble previous works on chaotic Hamiltonians that apply the Eigenstate Thermalization Hypothesis (ETH) to prove the fast mixing of dissipative dynamics \cite{ETH_thermalization_Chen21, PhysRevResearch.6.033147}.
In particular, in \cite{ETH_thermalization_Chen21}, the proposed algorithm implements a ``rounded'' Davies generator, yielding a physical Lindbladian that block-diagonalizes into components consisting of small-energy transitions. 
They propose their own version of ETH that relies on jump operators, for small Bohr frequencies $\omega$, having independent Gaussian-distributed entries. 
The assumption that these entries are independent for the result is very strong, allowing them to conclude that their jump operators are both local \emph{and} that distinct energy transitions are completely uncorrelated. 

Our work shows fast mixing unconditionally for quantumly easy Hamiltonians, replacing these local jumps and ETH assumption for the rounded Davies generator with 1-design jump operators for the newer CKG Lindbladian.
A similar ETH assumption to \cite{ETH_thermalization_Chen21} would also yield fast-mixing for the CKG Lindbladian with local jumps, but more generally some approach must be taken to characterize how correlations induced by implementing local jumps (in contrast with 1-design jump operators) can be mitigated to ensure fast-mixing for certain local Hamiltonians. 

It is an interesting future direction to extend our analysis to unconditionally establish fast mixing of chaotic Hamiltonians such as the SYK model \cite{SYK, kitaev2015SYK, kitaev2015simple} under the CKG Lindbladian. 

\section{Acknowledgments}
We thank Jiaoyang Huang, Luca Nashabeh, John Preskill, and Yongtao Zhan   for valuable discussions. 
We are grateful to Chi-Fang Chen for insightful discussions and proposing this project in its early stages.
MS is supported by AWS Quantum Postdoctoral Scholarship and funding from the National Science Foundation NSF CAREER award CCF-2048204.
Institute for Quantum Information and Matter is an NSF Physics Frontiers Center. AR acknowledges the Caltech SURF program and Zhuang Tang and Gebing Yi for their generous support. 

\section{Appendix}
\label{sec:appendix}

\subsection{Preliminaries}
We begin with a description of the Davies generator, which is the limit of the CKG Lindbladian as $\sigma_E \to 0$. 
This generator was developed from a physical approximation of an open thermalizing quantum system, but at low temperatures it is unphysical, and as a result in general it is hard to implement.
We then proceed to generalize the notions to the efficiently implementable CKG Lindbladian. 
\subsubsection{Davies generator}\label{sec:CKGdetails2}
In the description of the Davies generator for a given system $\vH$, there is a coherent term and there are Lindblad operators $\vA^a$.
The $\vA^a$ terms must appear in adjoint pairs in the construction of the Davies generator. 
The dissipative part of the  Lindbladian is expressed as follows:

\begin{align}
    \mathcal{L}_\beta[\cdot]= \sum_{a \in [M]} \int_{-\infty}^\infty \gamma(\omega) \left(\vA^a_\omega (\cdot) \vA_\omega^\dag - \frac{1}{2}\{\vA_\omega^\dag\vA_\omega, \cdot\}\right)d\omega, \label{eq:lindbladianForm}
    \end{align}
where $\vA^a_\omega$ is the Operator Fourier Transform (OFT) of jump operator $\vA^a$: 

\[\vA^a_\omega = \frac{1}{\sqrt{2\pi}}\int_{-\infty}^\infty e^{i\vH t }\vA^a e^{-i\vH t}e^{-i\omega t}dt.\]
The Davies' generator chooses Lindblad operators $\sqrt{\gamma(\omega)}\vA^a_\omega$, each of which selects the energy transitions, or Bohr frequencies, in $\vA^a$ that are precisely $\omega$. Because it requires certainty in energy, by Heisenberg's uncertainty principle of energy and time, in the general case simulating the evolution of the Davies generator efficiently is infeasible. 
In the above, $\gamma$ is some function satisfying $\gamma(\omega) = \gamma(\omega) = e^{-\beta\omega}\gamma(-\omega)$. 
The Lindblad operators are scaled by $\gamma(\omega)$ precisely to satisfy KMS detailed balance.
Since $A^a_\omega$ represents jumps with Bohr frequency $\omega$, the functional equation of $\gamma$ ensures a desired ratio of jumps with Bohr frequency $\omega$ and $-\omega$.  
We choose the Metropolis filter, $\gamma(\omega) = \min(1, e^{-\beta\omega})$, though another common filter $\gamma(\omega) = \frac{1}{1 + e^{-\beta\omega}}$ for ``Glauber dynamics'' could also be used for the same results. 

The Davies generator satisfies detailed balance with respect to $\vrho_\beta$, the thermal state. 
In some presentations of the Davies generator, it contains a coherent term $-i[\vH, \cdot]$. 
If this term is included, the generator does not satisfy detailed balance, so we do not follow this convention. 
However, the term does not affect the fixed point of the generator, since $\vrho_\beta$ commutes with $\vH$ and therefore $-i[\vH, \vrho_\beta] = 0$. 
\subsubsection{CKG Lindbladian}
The CKG Lindbladian is defined almost identically to the Davies generator, but is altered slightly so that it still obeys detailed balance, but is efficiently implementable. 

\[\mathcal{L}_\beta[\cdot] = \underbrace{-i[\vB, \cdot]}_{\text{coherent term}} + \sum_{a \in [M]} \int_{-\infty}^\infty \gamma(\omega) \left(\underbrace{\hat{\vA}^a(\omega) (\cdot) \hat{\vA}^a(\omega)^\dag}_{\text{transition term}} - \underbrace{\frac{1}{2}\{\hat{\vA}^a(\omega)^\dag\hat{\vA}^a(\omega), \cdot\}}_{\text{decay term}}\right)d\omega,\]
where $\hat{\vA}^a(\omega)$ is now the Gaussian-supported OFT of jump operator~$\vA^a$: 

\[\hat{\vA}^a(\omega) = \frac{1}{\sqrt{2\pi}}\int_{-\infty}^\infty e^{i\vH t }\vA^a e^{-i\vH t}e^{-i\omega t}f(t)dt.\]
To ensure that the jump operators do not have infinite precision in energy, a Gaussian supported OFT is performed instead to obtain $\hat{\vA}^a(\omega)$, which selects a Gaussian band energies of around $\omega$. 

Here, $f(t) = e^{-\sigma_E^2t^2}\sqrt{\sigma_E\sqrt{2/\pi}}$, with Fourier transform $\hat{f}(\omega) = \frac{1}{\sqrt{\sigma_E\sqrt{2\pi}}}\exp(-\frac{\omega^2}{4\sigma_E^2})$. As a result, the operator 
$\hat{\vA}^a(\omega)$ can be shown to be equal to $\sum_\nu \hat{f}(\omega - \nu)\vA^a_\nu$.
The function $f(t)$ was chosen so that its squared Fourier transform $\hat{f}^2(\omega)$ is a Gaussian with standard deviation $\sigma_E$, which features prominently in the Lindbladian (since it consists of quadratic terms in $\hat{\vA}^a(\omega)$). 
Taking $\sigma_E = \Theta(\beta^{-1})$ yields an efficient simulation algorithm with the assumption of a block-encoding of $\vH$ and a block-encoding for the jump operators $\sum_{a \in [M]} \ket{a} \otimes \vA^a$, so $\sigma_E$ is taken to be on the order of $\beta^{-1}$ in this paper. 

Since $\vA^a(\omega)$ is a noisy decomposition of $\vA^a$ into frequencies, it is not immediately clear whether there is a choice of function $\gamma(\omega)$ for which they can be recombined to achieve detailed balance.
Indeed, as shown in \cite{chen2023efficient}, there is!
The choice of $\gamma(\omega)$ is such that the the transition part of $\mathcal{L}_\beta$, the summand $\sum_{a \in [M]} \int_{-\infty}^\infty \gamma(\omega) \hat{\vA}^a(\omega) (\cdot) \hat{\vA}^a(\omega)^\dag d\omega$, still satisfies KMS detailed balance. 
\cite{chen2023efficient} proved that there is a unique choice of $\vB$, up to translation by a scalar, such that $-i[\vB, \cdot] - \frac{1}{2}\sum_{a \in [M]} \int_{-\infty}^\infty \gamma(\omega)  \{\hat{\vA}^a(\omega)^\dag\hat{\vA}^a(\omega), \cdot\} d\omega$ also satisfies detailed balance.
For the Davies generator, this coherent term $\vB$ is simply $0$ (or corresponds to a Lamb shift that commutes with the Hamiltonian), and the decay term by itself already satisfies detailed balance. 
$\vB$ can be expressed in general as:

\[\vB = \sum_{a \in [M]} \sum_{\nu_1, \nu_2} \frac{\tanh(-\beta(\nu_1 - \nu_2)/4)}{2i}(\vA^a_{\nu_2})^\dag\vA^a_{\nu_1}.\]

The choice of $\gamma$ for this algorithm, for which the filter is efficiently implementable, is $\gamma(\omega) = \exp\left(-\beta \max\left(\omega + \frac{\beta \sigma_E^2}{2}, 0\right)\right)$.
As $\sigma_E \to 0$ it converges to the Metropolis filter of the Davies generator. 
In particular, this $\gamma$ is precisely Metropolis filter for the Davies generator shifted by $\beta \sigma_E^2$.
The following lemma asserts that this is no coincidence.

\begin{lem}
If $\widetilde{\gamma}$ satisfies the functional equation $\widetilde{\gamma}(\omega) = \widetilde{\gamma}(-\omega)\exp(-\beta \omega)$, then the transition part of the CKG Lindbladian $\sum_a \int_{-\infty}^{\infty} \gamma(\omega)\left(\hat{\vA}^a(\omega)( \cdot) \hat{\vA}^a(\omega)^\dag\right)d \omega$ satisfies detailed balance for $\gamma(\omega) = \widetilde{\gamma}\left(\omega + \frac{\beta \sigma_E^2}{2}\right)$. 
\end{lem}
\begin{proof}
We prove the result for just one jump operators and its adjoint pair, since it extends by linearity to the full result. 
We can write the operator as follows, grouping together terms with opposite values of $\omega$ for the two jump operators:
\[\mathcal{L}_t = \int_{-\infty}^\infty \gamma(\omega) \hat{\vA}(\omega)(\cdot)\hat{\vA}(\omega)^\dag + \gamma(-\omega)\widehat{\vA^\dag}(-\omega)(\cdot)\widehat{\vA^\dag}(-\omega)^\dag d\omega.\]
By the identity $\hat{\vA}(\omega) = \sum_\nu \hat{f}(\omega - \nu)\vA_\nu$, note that $\widehat{\vA^\dag}(-\omega) = \sum_\nu \hat{f}(-\omega - \nu) (\vA^\dag)_{\nu} = \sum_\nu \hat{f}(-\omega + \nu) \vA_{\nu}$. We may therefore expand the previous expression to 
\begin{align*}
\mathcal{L}_t &= \sum_{\nu_1, \nu_2}\int_{-\infty}^\infty \gamma(\omega) \left(\hat{f}(\omega - \nu_1)\vA_{\nu_1}\right)(\cdot)\left(\hat{f}(\omega - \nu_2)\vA_{\nu_2}\right)^\dag + \gamma(-\omega)\left(\hat{f}(-\omega + \nu_1)\vA_{\nu_1}^\dag\right)(\cdot)\left(\hat{f}(-\omega + \nu_2)\vA_{\nu_2}^\dag\right)^\dag d\omega\\
&= \sum_{\nu_1, \nu_2}\left(\int_{-\infty}^\infty \hat{f}(\omega - \nu_1)\hat{f}(\omega - \nu_2)\gamma(\omega) d\omega\right)\vA_{\nu_1}(\cdot)\vA_{\nu_2}^\dag + \left(\int_{-\infty}^\infty \hat{f}(-\omega + \nu_1)\hat{f}(-\omega + \nu_2)\gamma(-\omega) d\omega\right)\vA_{\nu_1}^\dag(\cdot)\vA_{\nu_2}\\
&= \sum_{\nu_1, \nu_2}\theta(\nu_1, \nu_2)\vA_{\nu_1}(\cdot)\vA_{\nu_2}^\dag + \theta(-\nu_1, -\nu_2)\vA_{\nu_1}^\dag(\cdot)\vA_{\nu_2}.
\end{align*}

Throughout the rest of the paper, we denote 
\begin{align}
    \theta(\nu_1, \nu_2): = \int_{-\infty}^\infty \hat{f}(\omega - \nu_1)\hat{f}(\omega - \nu_2)\gamma(\omega) d\omega,\quad \text{and} \quad \alpha(\nu) := \theta(\nu,\nu).
\end{align}
Meanwhile, $\hat{f}(\omega - \nu)$ is the square root of a Gaussian distribution around $\nu$ with standard deviation $\sigma_E$.
This means the function $\alpha(\nu)$ is the convolution of $\gamma$ and a Gaussian of width $\sigma_E$. 
The function $\theta(\nu_1,\nu_2)$ has the property that $$\theta(\nu_1, \nu_2) = \alpha
\left(\frac{\nu_1+\nu_2}{2}\right)\exp\left(-\frac{(\nu_1 - \nu_2)^2}{8\sigma_E^2}\right).$$ 

\noindent
To verify detailed balance, we must prove that $\mathcal{L}_t[\cdot] = \vrho_\beta^{1/2}\mathcal{L}_\beta^\dag[\vrho_\beta^{-1/2} (\cdot) \vrho_\beta^{-1/2}]\vrho_\beta^{1/2}$. 
The adjoint of $\mathcal{L}_t$~is 
\[\mathcal{L}_t^\dag = \sum_{\nu_1, \nu_2}\theta(\nu_1, \nu_2)\vA_{\nu_1}^\dag(\cdot)\vA_{\nu_2} + \theta(-\nu_1, -\nu_2)\vA_{\nu_1}(\cdot)\vA_{\nu_2}^\dag.\]
We must therefore prove that 
\begin{align*}
\mathcal{L}_t &\stackrel{?}{=} \sum_{\nu_1, \nu_2}\theta(\nu_1, \nu_2)\vrho_\beta^{1/2}\vA_{\nu_1}^\dag\vrho_\beta^{-1/2}(\cdot)\vrho_\beta^{-1/2}\vA_{\nu_2}\vrho_\beta^{1/2} + \theta(-\nu_1, -\nu_2)\vrho_\beta^{1/2}\vA_{\nu_1}\vrho_\beta^{-1/2}(\cdot)\vrho_\beta^{-1/2}\vA_{\nu_2}^\dag\vrho_\beta^{1/2}\\
&= \sum_{\nu_1, \nu_2} \theta(\nu_1, \nu_2)\exp\left(\beta\frac{(\nu_1 + \nu_2)}{2}\right)\vA_{\nu_1}^\dag(\cdot)\vA_{\nu_2} + \theta(-\nu_1, -\nu_2)\exp\left(\beta\frac{-(\nu_1 + \nu_2)}{2}\right)\vA_{\nu_1}(\cdot)\vA_{\nu_2}^\dag.
\end{align*}
The second equality holds because $\vrho_\beta = \frac{\exp(-\beta vH)}{Z}$. Therefore, since $\vA_{\nu}^\dag$ only makes jumps with $\Delta E = -\nu$ and $\vA_\nu$ only makes jumps with $\Delta E = \nu$, it can be verified that $\vrho_\beta^{1/2} \vA_\nu \vrho_\beta^{-1/2} = \exp(-\frac{\beta}{2} \nu)$ and likewise $\vrho_\beta^{-1/2} \vA_\nu \rho_\beta^{1/2} = \exp(\frac{\beta}{2} \nu)$.

Inspecting the expression for $\mathcal{L}_t$, we conclude that to satisfy the above equality it is sufficient to have $\theta(\nu_1, \nu_2) = \theta(-\nu_1, \nu_2)\exp(\beta\frac{(\nu_1 + \nu_2)}{2})$ for any pair $\nu_1, \nu_2$.
In the case that $\nu = \nu_1 = \nu_2$, this is precisely the condition that $\alpha(\nu) = \alpha(-\nu)\exp(-\beta \nu)$. 
In fact, this functional equation for $\alpha$ is not only necessary but sufficient for the general case. 
Indeed, 
\begin{align*}
\theta(\nu_1, \nu_2) &= \alpha\left(\frac{\nu_1 + \nu_2}{2}\right)\exp\left(-\frac{\left(\nu_1 - \nu_2\right)^2}{8\sigma_E^2}\right)\\
&= \exp\left(-\beta\left(\frac{\nu_1+\nu_2}{2}\right)\right)\alpha\left(\frac{-(\nu_1 + \nu_2)}{2}\right)\exp\left(-\frac{\left((-\nu_1) - (-\nu_2)\right)^2}{8\sigma_E^2}\right)\\
&= \theta(-\nu_1, \nu_2)\exp\left(-\beta\left(\frac{\nu_1+\nu_2}{2}\right)\right).
\end{align*}
To conclude, we must prove that $\alpha$ satisfies $\alpha(\nu) = \alpha(-\nu)\exp(-\beta \nu)$, where $\alpha = \gamma * g$ and $g$ is a Gaussian with width $\sigma_E$ around 0. 
We may write:
\begin{align*}
\exp(-\beta \nu)\alpha(-\nu) &= \exp(-\beta \nu)\int_{-\infty}^\infty \gamma(\omega)g(-\nu - \omega)d\omega\\
&= \frac{\exp(-\beta \nu)}{\sigma_E\sqrt{2\pi}}\int_{-\infty}^\infty \gamma(\omega)\exp\left(-\frac{(\nu+\omega)^2}{2\sigma_E^2}\right)d\omega\\
&= \frac{\exp(-\beta \nu)}{\sigma_E\sqrt{2\pi}}\int_{-\infty}^\infty \gamma(-\omega - \beta \sigma_E^2)\exp\left(-\beta\left(\omega + \frac{\beta \sigma_E^2}{2}\right)\right)\exp\left(-\frac{(\nu+\omega)^2}{2\sigma_E^2}\right)d\omega\\
&= \frac{1}{\sigma_E\sqrt{2\pi}}\int_{-\infty}^\infty \gamma(-\omega - \beta \sigma_E^2)\exp\left[-\beta\left(\omega +\nu + \frac{\beta \sigma_E^2}{2}\right) - \frac{(\nu + \omega)^2}{2\sigma_E^2}\right]d\omega 
\end{align*}
where the third equality holds by utilizing the functional equation for $\widetilde{\gamma}$, which states that $\gamma(\omega - \frac{\beta \sigma_E^2}{2}) = \gamma(-\omega - \frac{\beta \sigma_E^2}{2})\exp\left(-\beta \omega\right)$ and therefore $\gamma(\omega) = \gamma(-\omega - \beta \sigma_E^2)\exp\left(-\beta \left(\omega + \frac{\beta \sigma_E^2}{2}\right)\right)$. Now, we express $\alpha$, writing

\begin{align*}
\alpha(\nu) &= \int_{-\infty}^\infty \gamma(\omega)g(\nu - \omega)d\omega\\
&= \frac{1}{\sigma_E\sqrt{2\pi}}\int_{-\infty}^\infty \gamma(\omega)\exp\left(-\frac{(\nu-\omega)^2}{2\sigma_E^2}\right)d\omega\\
&= \frac{1}{\sigma_E\sqrt{2\pi}}\int_{-\infty}^\infty \gamma(-\omega-\beta \sigma_E^2)\exp\left(-\frac{(\nu+\omega+\beta\sigma_E^2)^2}{2\sigma_E^2}\right)d\omega
\end{align*}
where in the third step we replace $\omega$ by $\omega - \beta \sigma_E^2$. By calculation, we see that $-\frac{(\nu+\omega +\beta\sigma_E^2)^2}{2\sigma_E^2} = -\beta\left(\omega +\nu + \frac{\beta \sigma_E^2}{2}\right) - \frac{(\nu + \omega)^2}{2\sigma_E^2}$, and therefore our expressions for $\alpha(\nu)$ and $\exp(-\beta \nu)\alpha(-\nu)$ are equal, completing the proof. 
\end{proof}
The CKG Lindbladian requires a choice of $\gamma$ for which $\alpha = \gamma * g$ satisfies the functional equation that was originally satisfied by $\gamma$ in the Davies generator. 
Fortunately, if $\tilde{\gamma}(\omega)$ is such a solution to the equation, like the Metropolis filter of the Davies generator, then $\tilde{\gamma}(\omega + \frac{\beta \sigma_E^2}{2}) * g$ does as well. 
For the rest of this paper, therefore, we use $\gamma$ to denote $\gamma(\omega) = \exp\left(-\beta \max\left(\omega + \frac{\beta \sigma_E^2}{2}, 0\right)\right)$, and $\alpha$ to denote $\gamma * g$. 

We also record a result about the operator norm of $\mathcal{L}_\beta$:
\begin{lem}\label{lem:operatornorm}
Consider the CKG Lindbladian $\mathcal{L}_\beta$ with temperature $\beta^{-1}$, using the Metropolis filter, and with jump operators $\vA^a$ for which $\sum_a \norm{\vA^{a\dag}\vA^a} \leq 1$. 
This Lindbladian satisfies
\begin{align}
    \norm{\mathcal{L}_\beta}_{\infty \to \infty} = O(\log(\beta\norm{\vH})),\label{eq:normLbound}
\end{align}
where $\norm{\cdot}_{\infty \to \infty}$ is the operator norm of $\mathcal{L}_\beta$, with respect to the operator norm on the input and output vector spaces. 
\end{lem}
\begin{proof}
Proposition B.2 in \cite{chen2023efficient} implies that the coherent term of the Lindbladian $B$ has operator norm $O(\log(\beta \norm{\vH}))$, which implies immediately that 
\[\norm{-i[\vB, \vrho]} \leq \norm{\vB\vrho}+\norm{\vrho\vB} \leq 2\norm{\vB}\norm{\vrho}.\]
In particular, $\norm{-i [\vB, \cdot]}_{\infty \to \infty} = O\left(\log(\beta \norm{\vH})\right)$. 

It remains to bound the transition and decay terms of the Lindbladian. 
The transition term of the Lindbladian can be bounded as
\begin{align}
    \left \lVert \sum_{a \in [M]} \int_{-\infty}^\infty \gamma(\omega) \hat{\vA}^a(\omega) \vrho \hat{\vA}^a(\omega)^\dag
d\omega  \right \rVert &\leq \left \lVert \rho \right \rVert \cdot \left \lVert \gamma \right  \rVert \cdot \left \lVert \sum_{a \in [M]} \int_{-\infty}^\infty \hat{\vA}^a(\omega) \hat{\vA}^a(\omega)^\dag
d\omega \right \rVert \nonumber\\
&\leq \int_{-\infty}^{\infty} |f(t)|^2 dt \cdot \left \lVert\sum_a \vA^{a\dag}\vA^{a}\right \rVert \leq 1.
\end{align}
In the first line, we applied Lemma K.1 from \cite{chen2023local}. To reach the second line, we used the fact that $\max_{\omega} |\gamma(\omega)| \leq 1$, applied Operator Parseval's identity from Lemma K.1 of \cite{chen2023local} or Proposition A.1 of \cite{chen2023QThermalStatePrep}, and utilized the normalization of the jump operators along with the Gaussian distribution.

A very similar calculation implies that the norm of the decay term can be bounded by
\begin{align}
    \left \lVert \sum_{a \in [M]} \int_{-\infty}^\infty \frac{1}{2} \gamma(\omega) \{\hat{\vA}^a(\omega)^\dag\hat{\vA}^a(\omega), \vrho\} d\omega \right \rVert \leq 1.
\end{align}
Combining these norm bounds on the  coherent, transition and decay terms implies \eqref{eq:normLbound}.

\end{proof}

\subsubsection{CKG Lindbladian in the energy basis}\label{sec:CKGinEigenbasis}
We consider a quantum system consisting of basis states $\ket{e_i}$ and a Hamiltonian $\vH$.
We choose some jump operators $\vA^a$ and denote $A^a_{lm} = \bra{l}\vA^a \ket{m}$. 
Notate the energy eigenstates as $\ket{j}$ with energy $E_j$. 
We independently calculate the three parts of the Lindbladian: the transition term $\mathcal{L}_t$, the decay term $\mathcal{L}_d$, and coherent term $\mathcal{L}_c$ so that $\mathcal{L}_\beta = \mathcal{L}_c + \mathcal{L}_t + \mathcal{L}_d$. 
First, as mentioned above, the OFT of the jump operator $\vA^a$ is:

\begin{align*}
\hat{\vA}^a(\omega) &= \frac{1}{\sqrt{2\pi}}\sum_{lm}A^a_{lm}\int_{-\infty}^\infty e^{i\vH t}\ketbra{l}{m}e^{-i\vH t}e^{i\omega t}f(t)dt\\
&= \frac{1}{\sqrt{2\pi}}\sum_{lm}A^a_{lm}\int_{-\infty}^\infty \ketbra{l}{m}e^{i\nu_{lm}t-i\omega t}f(t)dt\\
&= \sum_{lm}A^a_{lm}\hat{f}(\omega - \nu_{lm})\ketbra{l}{m}
\end{align*}
where $\nu_{lm} = E_l - E_m$. 
To represent superoperators as linear maps, we vectorize operators with respect to the basis of operators $\ketbra{m_1}{m_2}$. 
In particular, $\ket{\mathbf{m}}$ with $\mathbf{m} = (m_1, m_2)$ will notate the basis operator $\ketbra{m_1}{m_2}$. 
Now, we may expand $\mathcal{L}_t[N] = \sum_a\int_{-\infty}^\infty \gamma(\omega) \hat{\vA}^a(\omega) M \hat{\vA}^i(\omega)^\dag d\omega$ as

\begin{align*}
&\sum_a \int_{-\infty}^\infty \gamma(\omega)\left(\sum_{l_1m_1}A^a_{l_1m_1}\hat{f}(\omega - \nu_{l_1m_1})\ketbra{l_1}{m_1}\right)N\left(\sum_{l_2m_2}A^a_{l_2m_2}\hat{f}(\omega - \nu_{l_2m_2})\ketbra{l_2}{m_2}\right)^\dag d \omega=\\
&\sum_a \int_{-\infty}^\infty\gamma(\omega) \left(\sum_{l_1m_1}A^a_{l_1m_2}\hat{f}(\omega - \nu_{l_1m_1})\ketbra{l_1}{m_1}\right)N\left(\sum_{l_2m_2}\overline{A^a_{l_2m_2}}\hat{f}(\omega - \nu_{l_2m_2})\ketbra{m_2}{l_2}\right)d \omega=\\
    &\sum_{a, \mathbf{m}, \mathbf{l}} A^a_{l_1m_1}\overline{A^a_{l_2m_2}}\left(\int_{-\infty}^\infty \gamma(\omega)\hat{f}(\omega - \nu_{l_1m_1})\hat{f}(\omega - \nu_{l_2m_2})d\omega\right)(\ketbra{l_1}{m_1}N\ketbra{m_2}{l_2}).
\end{align*}
From the equation above, we may describe the elements of $\mathcal{L}_t$ as  \begin{align*}
\bra{\mathbf{l}}\mathcal{L}_t\ket{\mathbf{m}} &= \sum_a A^a_{l_1m_1}\overline{A^a_{l_2m_2}}\left(\int_{-\infty}^\infty \gamma(\omega)\hat{f}(\omega - \nu_{l_1m_1})\hat{f}(\omega - \nu_{l_2m_2})d\omega\right)\\
&= \sum_a A^a_{l_1m_1}\overline{A^a_{l_2m_2}}\theta(\nu_{l_1m_1}, \nu_{l_2m_2})
\end{align*}
with the notation $\theta(\nu_1, \nu_2)$ as defined previously. 
Notice that as $\sigma_E$ tends to zero, $\theta$ is only nonzero when its two inputs are equal.
The term $\bra{\mathbf{l}}\mathcal{L}_t\ket{\mathbf{m}}$ is therefore only nonzero if $\nu_{l_1m_1}= \nu_{l_2m_2}$, or equivalently, $\nu_{l_1l_2}= \nu_{m_1m_2}$. 
In the Davies limit, therefore, $\ketbra{m_1}{m_2}$ can only map to terms $\ketbra{l_1}{l_2}$ with the same Bohr~frequency.

Now, we calculate the decay part $\mathcal{L}_d$. We may expand $\mathcal{L}_d[N] = -\frac{1}{2}\sum_a \hat{\vA}^a(\omega)^\dag \hat{\vA}^a(\omega) N-\frac{1}{2}\sum_a N\hat{\vA}^a(\omega)^\dag \hat{\vA}^a(\omega)$ and deal with these terms individually. 
With the benefit of hindsight, we choose different indices for the terms of the jump operators, so that the expression will similarly simplify to yield $\bra{\mathbf{l}}\mathcal{L}_d\ket{\mathbf{m}}$. 
The first term is:

\begin{align*}
&-\frac{1}{2}\sum_a \int_{-\infty}^\infty \gamma(\omega)\left(\sum_{j_1m_2}A^a_{j_1m_2}\hat{f}(\omega - \nu_{j_1m_2})\ketbra{j_1}{m_2}\right)^\dag\left(\sum_{j_2l_2}A^a_{j_2l_2}\hat{f}(\omega - \nu_{j_2l_2})\ketbra{j_2}{l_2}\right)N d\omega = \\
&-\frac{1}{2}\sum_a \int_{-\infty}^\infty \gamma(\omega) \left(\sum_{j_1m_2}\overline{A^a_{j_1m_2}}\hat{f}(\omega - \nu_{j_1m_2})\ketbra{m_2}{j_1}\right)\left(\sum_{j_2l_2}A^a_{j_2l_2}\hat{f}(\omega - \nu_{j_2l_2})\ketbra{j_2}{l_2}\right) N d \omega= \\
&-\frac{1}{2} \sum_{a, l_2, m_2, j}\overline{A^a_{jm_2}}A^a_{jl_2}\left(\int_{-\infty}^\infty \gamma(\omega)\hat{f}(\omega - \nu_{jm_2})\hat{f}(\omega - \nu_{jl_2})d\omega\right)(\ketbra{l_2}{m_2}N).
\end{align*}
Similarly, the other term is:
\begin{align*}
&-\frac{1}{2}\sum_a \int_{-\infty}^\infty \gamma(\omega)\left(\sum_{j_1l_1}A^a_{j_1l_1}\hat{f}(\omega - \nu_{j_1l_1})\ketbra{j_1}{l_1}\right)^\dag\left(\sum_{j_2m_1}A^a_{j_2m_1}\hat{f}(\omega - \nu_{j_2m_1})\ketbra{j_2}{m_1}\right) N d\omega = \\
&-\frac{1}{2}\sum_a \int_{-\infty}^\infty \gamma(\omega) \left(\sum_{j_1l_1}\overline{A^a_{j_1l_1}}\hat{f}(\omega - \nu_{j_1l_1})\ketbra{l_1}{j_1}\right)\left(\sum_{j_2m_1}A^a_{j_2m_1}\hat{f}(\omega - \nu_{j_2m_1})\ketbra{j_2}{m_1}\right) N d \omega= \\
&-\frac{1}{2} \sum_{a, l_1, m_1, j}\overline{A^a_{jl_1}}A^a_{jm_1}\left(\int_{-\infty}^\infty \gamma(\omega)\hat{f}(\omega - \nu_{jl_1})\hat{f}(\omega - \nu_{jm_1})d\omega\right)(\ketbra{l_1}{m_1}N).
\end{align*}
From this expression, we find that 
\begin{align}
\bra{\mathbf{l}}\mathcal{L}_d\ket{\mathbf{m}} =-\frac{1}{2} \left( \delta_{l_1m_1}\sum_{a,j}\overline{A^a_{jm_2}}A^a_{jl_2}\theta(\nu_{jm_2}, \nu_{jl_2}) + \delta_{l_2m_2}\sum_{i, j}\overline{A^a_{jl_1}}A^a_{jm_1}\theta(\nu_{jl_1}, \nu_{jm_1})\right).
\end{align}

Notice that here in the limit of $\sigma_E \to 0$, in both summands $E_{m_2} = E_{l_2}$ and $E_{m_1} = E_{l_1}$ for the $\theta$ term to be nonzero, so again $\nu_{l_1l_2} = \nu_{m_1m_2}$. 
This reasoning demonstrates that the Davies generator is block-diagonal, where each block corresponds to a Bohr frequency $\nu$. 
If the Bohr frequencies $\nu_{ij}$ for $i \neq j$ are all distinct, then the Davies generator consists of a ``classical'' block for the basis elements $\ket{m, m}$, and the rest of the Lindbladian is diagonal. 
This Lindbladian can be understood as a combination of an evolution on the diagonal (a classical Markov chain), and the damping of off-diagonal terms (a dephasing channel).  
\\[0.5cm]
Finally, we write the form of the coherent part, $\mathcal{L}_c$. 
It appears very similar to $\mathcal{L}_d$, since the commutator is taken with $B=  \sum_{a \in [M]} \sum_{\nu_1, \nu_2} \frac{\tanh(-\beta(\nu_1 - \nu_2)/4)}{2i}(\vA^a_{\nu_2})^\dag\vA^a_{\nu_1}$, similarly to the anticommutator with $ -\frac{1}{2}\sum_{a \in [M]} (\vA^a)^\dag\vA^a = -\frac{1}{2}\sum_{a \in [M]} \sum_{\nu_1, \nu_2} (\vA^a_{\nu_2})^\dag\vA^a_{\nu_1}$ in the decay term. 

\begin{align*}
\bra{\mathbf{l}}\mathcal{L}_c\ket{\mathbf{m}} =\frac{1}{2} \Biggl( &\delta_{l_1m_1}\tanh(\beta (\nu_{m_2l_2})/4)\sum_{i,j}\overline{A^i_{jm_2}}A^i_{jl_2}\theta(\nu_{jm_2}, \nu_{jl_2}) - \\ 
&\delta_{l_2m_2}\tanh(\beta (\nu_{l_1m_1})/4)\sum_{i, j}\overline{A^i_{jl_1}}A^i_{jm_1}\theta(\nu_{jl_1}, \nu_{jm_1})\Biggr ).
\end{align*}

Summing the descriptions of $\mathcal{L}_t$, $\mathcal{L}_d$, and $\mathcal{L}_c$ gives the full Lindbladian. 
Collecting all the terms obtained above, we find the below formulae. Note that they can all be taken in terms of $\alpha(\nu) = \theta(\nu, \nu)$ using the identity $\theta(\nu_1, \nu_2) = \alpha(\frac{\nu_1+\nu_2}{2})\exp\left(-\frac{(\nu_1 - \nu_2)^2}{8\sigma_E^2}\right)$.
\begin{align}
&\bra{\mathbf{l}}\mathcal{L}_t\ket{\mathbf{m}} = \sum_a A^a_{l_1m_1}\overline{A^a_{l_2m_2}}\theta(\nu_{l_1m_1}, \nu_{l_2m_2}),\nonumber\\
&\bra{\mathbf{l}}\mathcal{L}_d\ket{\mathbf{m}} =-\frac{1}{2} \left( \delta_{l_1m_1}\sum_{a,j}\overline{A^a_{jm_2}}A^a_{jl_2}\theta(\nu_{jm_2}, \nu_{jl_2}) + \delta_{l_2m_2}\sum_{a, j}\overline{A^a_{jl_1}}A^a_{jm_1}\theta(\nu_{jl_1}, \nu_{jm_1})\right),\nonumber\\
&\bra{\mathbf{l}}\mathcal{L}_c\ket{\mathbf{m}} =\frac{1}{2} \Biggl( \delta_{l_1m_1}\tanh(\beta \nu_{m_2l_2}/4)\sum_{a,j}\overline{A^a_{jm_2}}A^a_{jl_2}\theta(\nu_{jm_2}, \nu_{jl_2}) - \nonumber\\
&\hspace{2.3cm}\delta_{l_2m_2}\tanh(\beta \nu_{l_1m_1}/4)\sum_{a, j}\overline{A^a_{jl_1}}A^a_{jm_1}\theta(\nu_{jl_1}, \nu_{jm_1})\Biggr ). \label{eq:graphlindbladian}
\end{align}

It is worth noting that restricting to the indices ($m, m$), the coherent term is zero, since in both nonzero terms $l_1 = m_1 = m_2 = l_2$ and therefore the $\tanh$ terms vanish. 
The transition and decay parts of the Lindbladian simplify to:
\begin{align*}
&\bra{\mathbf{l}}\mathcal{L}_t\ket{\mathbf{m}} = \sum_a |A^a_{lm}|^2\alpha(\nu_{lm}),\\
&\bra{\mathbf{l}}\mathcal{L}_d\ket{\mathbf{m}} =-\delta_{lm}  \sum_{a,j}|A^a_{jm}|^2\alpha(\nu_{jm}).
\end{align*}
This restriction of the Lindbladian is therefore a classical Markov chain generator, as the columns sum to zero and the off-diagonal terms are nonnegative. 
We refer to this as the classical part of the Lindbladian—it represents the internal dynamics of the diagonal of the state. 
This restriction als has stationary state $\rho_{jj}$—indeed, the operator satisfies classical detailed balance, precisely because $\alpha$ satisfies the functional equation $\alpha(\nu) = \alpha(-\nu)\exp(-\beta \nu)$.
The Lindbladian is not necessarily block-diagonal (i.e., the classical restriction is not necessarily a block), but when it is, this represents a decoupling of the classical dynamics of the diagonal, that tends to the Gibbs distribution over the energy eigenstates, from the dephasing that takes off-diagonal elements to zero as the state converges to the Gibbs state.

\subsection{Graph-Local Jumps for Cyclic Graphs}\label{cyclicgapnproof}
In this section we prove Theorem \ref{thm:cyclicgapn_intro}. 
Consider a cyclic graph with $n$ vertices with adjacency matrix $\vH$, and eigenvectors $\ket{j}$.

The eigenbasis of a cyclic graph consists of vectors $\ket{j} = n^{-1/2}\sum_a \zeta_n^{-aj}\ket{e_a}$ with eigenvalues $2\cos\left(\frac{2 \pi j}{n} \right)$. 
The jump operators on the graph are chosen to be graph-local $\vA^a = n^{-1/2}\ketbra{e_a}{e_a}$, and therefore have coefficients $A^a_{lm} = n^{-3/2}\zeta_n^{a(l-m)}$. 
Now we observe that $\sum_a \zeta_n^{a(i-j)} = n\delta_{ij}$. We therefore have the relation that

\[\sum_a A_{l_1m_1}^a \overline{A_{l_2m_2}^a} = n^{-3}\sum_a \zeta_n^{a((l_1 - m_1)-(l_2-m_2))} = n^{-2}\delta_{(l_1-m_1)(l_2 - m_2)}.\]
We compute the components of the Lindbladian with these jump operators:
\begin{align*}
\bra{\mathbf{l}}\mathcal{L}_t\ket{\mathbf{m}}& = \sum_a A^a_{l_1m_1}\overline{A^a_{l_2m_2}}\theta(\nu_{l_1m_1}, \nu_{l_2m_2})
\bra{\mathbf{l}}\mathcal{L}_t\ket{\mathbf{m}}
\\&= n^{-2}\theta(\nu_{l_1m_1}, \nu_{l_2m_2})\delta_{(l_1 - m_1)(l_2 - m_2)}.
\end{align*}
With a similar computation, we obtain
\begin{align*}
\bra{\mathbf{l}}\mathcal{L}_d\ket{\mathbf{m}} &=-\frac{1}{2} \left( \delta_{l_1m_1}\sum_{a,j}\overline{A^a_{jm_2}}A^a_{jl_2}\theta(\nu_{jm_2}, \nu_{jl_2}) + \delta_{l_2m_2}\sum_{i, j}\overline{A^a_{jl_1}}A^a_{jm_1}\theta(\nu_{jl_1}, \nu_{jm_1})\right)\\
&= -\frac{n^{-2}}{2}\delta_{l_1m_1}\delta_{l_2m_2} (\theta(\nu_{jm_2}, \nu_{jl_2})+\theta(\nu_{jl_1}, \nu_{jm_1}))\\
&= -\frac{n^{-2}}{2}\delta_{l_1m_1}\delta_{l_2m_2} (\alpha(\nu_{jm_2})+\alpha(\nu_{jm_1})).
\end{align*}
Finally, we find that the coherent part vanishes:
\begin{align*}
\bra{\mathbf{l}}\mathcal{L}_c\ket{\mathbf{m}} &=\frac{1}{2} \Biggl( \delta_{l_1m_1}\tanh(\beta (\nu_{m_2l_2})/4)\sum_{i,j}\overline{A^i_{jm_2}}A^i_{jl_2}\theta(\nu_{jm_2}, \nu_{jl_2}) - \\ 
&\hspace{1.1cm}\delta_{l_2m_2}\tanh(\beta (\nu_{l_1m_1})/4)\sum_{i, j}\overline{A^i_{jl_1}}A^i_{jm_1}\theta(\nu_{jl_1}, \nu_{jm_1})\Biggr ) \\&= 
\frac{n^{-2}}{2}\delta_{l_1m_1}\delta_{l_2m_2}(\tanh(\beta(\nu_{m_2l_2})/4)\theta(\nu_{jm_2}, \nu_{jl_2})+\tanh(\beta(\nu_{l_1m_1})/4)\theta(\nu_{jl_1}, \nu_{jm_1}))\\&= 
0.
\end{align*}
The above formulae imply that the Lindbladian is block diagonal in the eigenbasis.
The coherent term vanishes and the decay term is fully diagonal.
Setting $k = m_1 - m_2$ and $k' = l_1 - l_2$, the transition term $\bra{\mathbf{l}}\mathcal{L}_t\ket{\mathbf{m}}$ is nonzero only if $k = k'$, due to the factor $\delta_{(l_1 - m_1)(l_2 - m_2)} = \delta_{(l_1 - l_2)(m_1 - m_2)}$. 
There is therefore one block corresponding to each $k$, which we will denote $\mathcal{L}^k$. 
The block for $k=0$ is the classical block of the Markov chain on the diagonal entries of the state. 
Finding its spectral gap, and then lower bounding the eigenvalues of the remaining $n-1$ blocks, yields a bound for the spectral gap of the Lindbladian. 

\subsubsection{Classical block}
We will show that the spectral gap of the classical block is asymptotically $\Omega(n^{-1})$. We have that
\[\bra{\mathbf{l}}\mathcal{L}^0_t\ket{\mathbf{m}} = \frac{1}{n^{2}}\delta_{(l_1 - m_1)(l_2 - m_2)}\theta(\nu_{l_1m_1}, \nu_{l_2m_2}) = \frac{1}{n^{2}}\alpha(\nu_{lm}).\] 
Secondly, \[\bra{\mathbf{l}}\mathcal{L}^0_d\ket{\mathbf{m}} = -\frac{1}{2n^{2}} \left( \delta_{l_1m_1}\delta_{l_2m_2}\left(\sum_{j} \alpha(\nu_{jm_2}) + \sum_{j} \alpha(\nu_{jm_1}) \right)\right) = -\frac{1}{n^{2}}\delta_{lm} \sum_j \alpha(\nu_{jm}).\]
It is evident that the columns sum to zero and the off-diagonal terms are nonnegative, so the matrix can be considered a Markov chain generator. 

We will repeatedly use the canonical path bound, a standard technique in the theory of classical Markov chains, to establish lower bounds on their spectral gaps.
The canonical path lemma states that, fixing a ``canonical'' path between each pair of vertices on a graph, a corresponding spectral graph bound can be obtained on the Markov chain.
For our purposes, every value of the Markov chain generator is nonzero, so it suffices to consider the canonical path between any two vertices to be the edge joining them. 
The lemma then simplifies to the following statement:
\begin{lem}\label{lem:canonpath}
Say $L^0$ is a Markov chain generator with stationary state $\sigma$. Then, the spectral gap satisfies the following bound:

\[\lambda \geq  \min_{(l, m)} \frac{L^0_{lm}}{\sigma_l}.\]
\end{lem}
\noindent
Applying this bound in this case, and noting that the stationary state of this Markov chain is $\rho_{ll}$, we obtain the lower bound
\begin{align}
\lambda \geq \min_{l\neq m} \frac{\alpha(\nu_{lm})n^{-2}}{\rho_{ll}}. 
\end{align}
The first equality holds because every canonical path is length 1, so the only path containing the edge $(l, m)$ is $\gamma_{lm}$.
We may upper bound $\rho_{ll}$ with $\rho_{ll} \leq \frac{e^{2\beta }}{\sum_{i} E_i} \leq n^{-1} \frac{e^{2\beta}}{e^{-2\beta}} = n^{-1}e^{4\beta}$. 
Moreover, $|\nu_{lm}| \leq 4$ since all energies lie in $[-2, 2]$, so $\alpha$ is bounded below by a positive constant $C$ that is independent of $n$. 
We conclude that $\lambda \geq \min_{l\neq m} \frac{\alpha(\nu_{lm})n^{-2}}{\rho_{ll}} \geq \frac{Cn^{-2}}{e^{4\beta}n^{-1}} = \Omega(n^{-1})$, as desired. 

\subsubsection{Non-classical blocks}
We utilize the Gershgorin bound on the columns of the $k$th block of $-\mathcal{L}_\beta$, $-\mathcal{L}^k$, which states that the eigenvalues of $\mathcal{L}^k$ must be larger than $\min_\mathbf{m} \left[\braket{\mathbf{m}| \mathcal{L}^k|\mathbf{m}} - \sum_{\mathbf{l} \neq \mathbf{m}}|\braket{\mathbf{l}|\mathcal{L}^k|\mathbf{m}}|\right]$. The approach we follow is very similar to the one outlined in \cite{temmedaviesgap2013} for the Davies generator. 
For a fixed block $k>0$, we now evaluate the bound, first establishing the following lemma:
\begin{lem}
For sufficiently large energy resolution $\sigma_E > 0$ such that for each non-classical block of the CKG Lindbladian $\mathcal{L}_\beta$ specified above, $\braket{\mathbf{m}| \mathcal{L}^k|\mathbf{m}} - \sum_{\mathbf{l} \neq \mathbf{m}}|\braket{\mathbf{l}|\mathcal{L}^k|\mathbf{m}}| \geq \frac{C}{n^{2}}\sum_{\mathbf{l}} (\nu_{l_1m_1} - \nu_{l_2m_2})^2$ for some constant $C >0$. 
\end{lem}
\begin{proof}
We first expand the term of interest:
\begin{align}
&\braket{\mathbf{m}| \mathcal{L}^k|\mathbf{m}} - \sum_{\mathbf{l} \neq \mathbf{m}}|\braket{\mathbf{l}|\mathcal{L}^k|\mathbf{m}}|\nonumber \\&= \braket{\mathbf{m}|\mathcal{L}^k_t|\mathbf{m}}+\braket{\mathbf{m}|\mathcal{L}^k_d|\mathbf{m}}-\sum_{\mathbf{l} \neq \mathbf{m}} |\braket{\mathbf{l}|\mathcal{L}^k_t|\mathbf{m}}|\nonumber
\\&=  \frac{1}{n^{2}}\sum_{\mathbf{l}}\left(\frac{\alpha(\nu_{l_1m_1}) + \alpha(\nu_{l_2m_2})}{2} - \alpha\left(\frac{\nu_{l_1m_1}+\nu_{l_2m_2}}{2}\right)\right)\nonumber\\
&+\frac{1}{n^{2}}\sum_{\mathbf{l}}\alpha\left(\frac{\nu_{l_1m_1}+\nu_{l_2m_2}}{2}\right)\left(1-\exp\left(\frac{-(\nu_{l_1m_1} - \nu_{l_2m_2})^2}{8\sigma_E^2}\right)\right) \nonumber.
\end{align}
Note that the sums over $\mathbf{l}$ above are only for the rows in the $k$th block, i.e. $l_1 - l_2 = k$. 
Now, since $1-\exp(-x) \geq \frac{x}{2}$ for $x \leq 1$, we have that for $\vert \nu_{l_1m_1} - \nu_{l_2m_2}\vert \leq \sqrt{8\sigma_E^2} = \sqrt{8}\sigma_E$, \[1-\exp\left(\frac{-(\nu_{l_1m_1} - \nu_{l_2m_2})^2}{8\sigma_E^2}\right) \geq \frac{(\nu_{l_1m_1} - \nu_{l_2m_2})^2}{16\sigma_E^2},\]
and therefore 
\[\frac{1}{n^2}\alpha\left(\frac{\nu_{l_1m_1}+\nu_{l_2m_2}}{2}\right)\left(1-\exp\left(\frac{-(\nu_{l_1m_1} - \nu_{l_2m_2})^2}{8\sigma_E^2}\right) \right)\geq \frac{M}{n^2}\frac{(\nu_{l_1m_1} - \nu_{l_2m_2})^2}{16\sigma_E^2}.\]
The values of $\nu$ are all bounded in the range $[-2, 2]$ so $\alpha$ has some constant lower bound, which we term $M>0$ above. Moreover, $\alpha$ is smooth, being the convolution of the smooth Gaussian with another function (the Metropolis-like filter). 
We therefore also have that 
\[\frac{1}{2}\left(\alpha(\nu_{l_1m_1})+\alpha(\nu_{l_2m_2})\right) - \alpha\left(\frac{(\nu_{l_1m_1}+\nu_{l_1m_1})}{2}\right) \geq -C'(\nu_{l_1m_1} - \nu_{l_2m_2})^2\]
for some $C' >0$. 
Setting $\sigma_E$ to a sufficiently small constant that $\frac{1}{16\sigma_E^2} \geq 2C'$, we may lower bound the expression of interest, taking only the summands $\mathbf{l'}$ with $\vert \nu_{l_1m_1} - \nu_{l_2m_2}\vert \leq \sqrt{8}\sigma_E$:
\begin{align}
\frac{1}{n^{2}}\left(\frac{1}{16\sigma_E^2} - C'\right)\sum_{\mathbf{l'}} (\nu_{l'_1m_1}-\nu_{l'_2m_2})^2  \geq \frac{C'}{n^{2}}\sum_{\mathbf{l'}} (\nu_{l'_1m_1}-\nu_{l'_2m_2})^2.
\end{align}
For $\vert \nu_{l_1m_1} - \nu_{l_2m_2}\vert \geq \sqrt{8}\sigma_E$, 
\[\exp\left(\frac{-(\nu_{l_1m_1} - \nu_{l_2m_2})^2}{8\sigma_E^2}\right) \leq \exp(-1) \leq \frac{1}{2}.\]
Note that since the Metropolis-like filter is monotone, so is the Gaussian-blurred $\alpha$, and therefore  $\alpha(\nu_{l_1m_1})+\alpha(\nu_{l_2m_2}) - \alpha\left(\frac{\nu_{l_1m_1}+\nu_{l_2m_2}}{2}\right) \geq M$, where $M$ is again the constant lower bound to $\alpha$ on the region $[-2, 2]$. 
For the summands for these values of $\mathbf{l}$, therefore, we obtain the lower bound of
\begin{align*}
\frac{1}{n^{2}}\left(\frac{\alpha(\nu_{l_1m_1}) + \alpha(\nu_{l_2m_2})}{2} - \alpha\left(\frac{\nu_{l_1m_1}+\nu_{l_2m_2}}{2}\right)\exp\left(\frac{-(\nu_{l_1m_1} - \nu_{l_2m_2})^2}{8\sigma_E^2}\right)\right) \geq \frac{M}{2n^{2}}.
\end{align*}
This value is at most a constant factor smaller than $\frac{1}{n^{2}}(\nu_{l_1m_1}-\nu_{l_2m_2})^2$, since $(\nu_{l_1m_1} - \nu_{l_2m_2})^2$ is upper bounded by a constant. 
As a result, we may finally assert that there exists a constant $C$ such that 
\[\braket{\mathbf{m}| \mathcal{L}^k|\mathbf{m}} - \sum_{\mathbf{l} \neq \mathbf{m}}|\braket{\mathbf{l}|\mathcal{L}^k|\mathbf{m}}| \geq \frac{C}{n^{2}}\sum_\mathbf{l} (\nu_{l_1m_1}-\nu_{l_2m_2})^2.\]
\end{proof}
Using the above lemma, we can now establish a lower bound on the spectral gap. 
Note that $\nu_{l_1m_1} - \nu_{l_2m_2} = \nu_{l_1l_2} - \nu_{m_1m_2} = (\cos(\frac{2\pi l_1}{n})-\cos(\frac{2\pi l_2}{n})) - (\cos(\frac{2\pi m_1}{n})-\cos(\frac{2\pi m_2}{n}))$. 
This value is equal to $\nu_{l_1m_1} - \nu_{l_2m_2} = \sin(\frac{\pi k}{n} )(\sin(\frac{\pi (l_1 + l_2)}{n} ) - \sin(\frac{\pi (m_1 + m_2)}{n} ))$. 
There is a constant fraction of choices of each $l_1$ such that $\left|\sin\left(\frac{\pi (l_1 + l_2)}{n} \right)\right| \geq 0.5$ and such that the sign of $\sin\left(\frac{\pi (l_1 + l_2)}{n} \right)$ is opposite to the sign of $\sin\left(\frac{\pi (m_1 + m_2)}{n} \right)$, so their difference has magnitude at least $0.5$.
Then, taking this fraction to be $f$:
\[\frac{C}{n^{2}}\sum_{\mathbf{l}} (\nu_{l_1m_1} - \nu_{l_2m_2})^2 \geq \frac{fC}{4n} \left|\sin\left(\frac{\pi k}{n}\right)
\right|^2 \]
The minimum value of this for any $k$ is $\Omega(n^{-3})$.
\subsubsection{Final bound}\label{sec:ProofProductCycle}
Now we prove Theorem~\ref{thm:cyclicgapn_intro}.
As shown above, $\mathcal{L}_\beta$ is a block diagonal matrix, with one classical block and $n-1$ non-classical blocks. 
The classical block is a Markov chain generator, and using the canonical path lemma the spectral gap was shown to be $\Omega(n^{-1})$. 
Meanwhile, the $k$th non-classical block has a minimum eigenvalue that is $\Omega(n^{-3})$, so the spectral gap of the full Lindbladian is $\Omega(n^{-3})$. 
\\[0.5cm]
To prove the upper bound on the spectral gap, we consider the row vector $v$ of length $n^{2}$, that is $1$ on the indices that correspond to the block $k=1$, and 0 elsewhere. 
As an operator, it takes the value 1 on one offdiagonal with a fixed $l_1 - l_2 =k$.
When calculating $(v\mathcal{L}_\beta)_\mathbf{m} = (\mathcal{L}_\beta^\dag v)_\mathbf{m}$, we obtain the same formula as found in the Gershgorin bound calculation above—indeed, the values on the diagonal are all positive, while the off-diagonal values are negative:
\begin{align}
(\mathcal{L}_\beta^\dag v)_\mathbf{m}=&\frac{1}{n^{2}}\sum_{\mathbf{l}}\left(\frac{\alpha(\nu_{l_1m_1}) + \alpha(\nu_{l_2m_2})}{2} - \alpha\left(\frac{\nu_{l_1m_1}+\nu_{l_2m_2}}{2}\right)\right)\nonumber \\
+ &\frac{1}{n^{2}}\sum_{\mathbf{l}}\alpha\left(\frac{\nu_{l_1m_1}+\nu_{l_2m_2}}{2}\right)\left(1-\exp\left(\frac{-(\nu_{l_1m_1} - \nu_{l_2m_2})^2}{8\sigma_E^2}\right)\right).
\end{align}
The previous lower bound shows that each of these values is nonnegative.
Since $l_1 - l_2 = m_1-m_2 = 1$, $\nu_{l_1m_1} - \nu_{l_2m_2} = \nu_{l_1m_1} - \nu_{l_2m_2} $ is $O(n^{-1})$. 
The first term in the expression, since it is composed of summands that are second differences in $\alpha$, is $n$ terms that are $O(n^{-2})$ scaled by $n^{-2}$—it is therefore $O(n^{-3})$. 
The summands of the second term can be similarly estimated to be $O(n^{-2})$, so it is also $O(n^{-3})$. 
The terms of $(\mathcal{L}_\beta^\dag v)_\mathbf{m}$ are therefore nonnegative and are at most $Cn^{-3}$ for some constant $C$. 

To prove the upper bound, we make use of the inner product $\langle \ ,\ \rangle_{\rho_\beta^{-1}}$ with respect to which $\mathcal{L}_\beta^\dag$ is self-adjoint. 
Note that $\langle \ketbra{i_1}{i_2}, \ketbra{j_1}{j_2} \rangle_{\rho_\beta^{-1}} = \delta_{i_2j_1}\delta_{i_1j_2} (\rho_\beta)_{i_1i_1}^{1/2}(\rho_\beta)_{i_2i_2}^{1/2} \geq 0$. 
Hence, when 
 $\langle \ , \ \rangle_{\rho_\beta^{-1}}$ is expressed in the energy basis as $vMw$ for a matrix $M$, $M$ has nonnegative elements. 
 We therefore may upper bound $\langle (\mathcal{L}_\beta^\dag v), v \rangle_{\rho_\beta^{-1}}$ by $Cn^{-3}\langle v, v \rangle$, since the coefficients of $(\mathcal{L}_\beta^\dag v)$ and $v$ are nonnegative and $(\mathcal{L}_\beta^\dag v)$ is dominated by $Cn^{-3}v$ for some $C>0$. 
 We conclude that $\frac{\langle (\mathcal{L}_\beta^\dag v), v \rangle_{\rho_\beta^{-1}}}{\langle v, v \rangle_{\rho_\beta^{-1}}}$ is $O(n^{-3})$. 
 Since $\mathcal{L}_\beta^\dag$ is self-adjoint with respect to this inner product, we obtain that $v$ has Rayleigh quotient $O(n^{-3})$. 
 $v$ is also orthogonal to $I$, since $\langle v, \rho_\beta\rangle_{\rho_\beta^{-1}} = \tr(v\rho_\beta^{1/2}I\rho_\beta^{1/2}) = \tr(v\rho_\beta) = 0$, where the last equality holds since as an operator $v$ is zero along the diagonal. 
 $v$ has no overlap with $I$, the fixed point of $\mathcal{L}_\beta^\dag$, and therefore its Rayleigh quotient is an upper bound on the spectral gap.
 The spectral gap must therefore also be $O(n^{-3})$.

\subsection{Bounded Degree Systems with 1-Design Jumps}\label{sec:BoundedProof}
In this section we prove Lemma \ref{lem:boundedaveragelindbladian} and Theorem \ref{thm:constgapbounded_intro}, demonstrating an improvement over the result in Theorem \ref{thm:cyclicgapn_intro} for local jumps in cyclic graphs. 

\subsubsection{Proof of Lemma \ref{lem:boundedaveragelindbladian}}\label{pf:boundedaveragelindbladian}
To prove Lemma \ref{lem:boundedaveragelindbladian}, we make use of the expressions~\eqref{eq:graphlindbladian} for the transition, decay, and coherent parts of a general Lindbladian, but with simply one Haar random jump (or equivalently any 1-design since the second moments of the operators are equal). The transition term is 
\[\bra{\mathbf{l}}\mathcal{L}_t\ket{\mathbf{m}} = A_{l_1m_1}\overline{A_{l_2m_2}}\theta(\nu_{l_1m_1}, \nu_{l_2m_2}).
\]
The expectation of the product $A_{l_1m_1}\overline{A_{l_2m_2}}$ is zero if $l_1 \neq l_2$ or $m_1 \neq m_2$. 
The expectation of the norm squared of an element, on the other hand, is $n^{-1}$. 
So we obtain
\begin{align}
\mathbb{E}\left[\bra{\mathbf{l}}\mathcal{L}_t\ket{\mathbf{m}}\right] &= \frac{\delta_{l_1l_2}\delta_{m_1m_2}}{n}\theta(\nu_{l_1m_1}, \nu_{l_2m_2}) \\
&= \frac{\delta_{l_1l_2}\delta_{m_1m_2}}{n}\alpha(\nu_{l_1m_1}).
\end{align}
Meanwhile, 
\begin{align*}
\mathbb{E}\left[\bra{\mathbf{l}}\mathcal{L}_d\ket{\mathbf{m}}\right] &=-\frac{1}{2} \left( \delta_{l_1m_1}\sum_{j}\mathbb{E}\left[\overline{A_{jm_2}}A_{jl_2}\right]\theta(\nu_{jm_2}, \nu_{jl_2}) + \delta_{l_2m_2}\sum_{ j}\mathbb{E}\left[\overline{A_{jl_1}}A_{jm_1}\right]\theta(\nu_{jl_1}, \nu_{jm_1})\right)\\
&= -\frac{1}{2} \left( \frac{\delta_{l_1m_1}\delta_{l_2m_2}}{n}\sum_{j}\theta(\nu_{jm_2}, \nu_{jm_2}) + \frac{\delta_{l_1m_1}\delta_{l_2m_2}}{n}\sum_{ j}\theta(\nu_{jm_1}, \nu_{jm_1})\right)\\
&= -\frac{1}{2} \left( \frac{\delta_{l_1m_1}\delta_{l_2m_2}}{n}\sum_{j}(\alpha(\nu_{jm_2})+\alpha(\nu_{jm_1}))\right), 
\\
\mathbb{E}\left[\bra{\mathbf{l}}\mathcal{L}_c\ket{\mathbf{m}}\right] &=\frac{1}{2} \Biggl( \delta_{l_1m_1}\tanh(\beta (\nu_{m_2l_2})/4)\sum_{j}\mathbb{E}\left[\overline{A_{jm_2}}A_{jl_2}\right]\theta(\nu_{jm_2}, \nu_{jl_2}) - \\ 
&\hspace{1.1cm}\delta_{l_2m_2}\tanh(\beta (\nu_{l_1m_1})/4)\sum_{ j}\mathbb{E}\left[\overline{A_{jl_1}}A_{jm_1}\right]\theta(\nu_{jl_1}, \nu_{jm_1})\Biggr ) = 0.
\end{align*}
The final Lindbladian $\mathcal{L}_\mu$ is therefore completely diagonal except for a ``classical block'' $\mathcal{L}^0$ of indices $\ket{(m, m)}$, whose off-diagonal terms are populated by the elements of $\mathcal{L}_t$. 
The spectral gap of this Lindbladian is therefore the minimum of the values along the diagonal, which are all positive, and the spectral gap of the classical block. 

As in \ref{cyclicgapnproof}, we use the Lemma \ref{lem:canonpath} to bound the spectral gap of the classical block.
Using this lemma, with the canonical path being the edge between a pair of vertices, we obtain the bound 
\begin{align}
\lambda \geq  \min_{l\neq m} \frac{\alpha(\nu_{lm})n^{-1}}{\rho_{ll}}.
\end{align}
We may upper bound $\rho_{ll}$ with $\rho_{ll} \leq \frac{e^{-E_\text{min} \beta }}{\sum_{i} E_i} \leq n^{-1} \frac{e^{-E_\text{min}\beta}}{e^{-E_\text{max}\beta}} = n^{-1}e^{-O(1)} = O(n^{-1})$ due to the fact that $\beta\norm{\vH} = O(1)$. 
Similarly, since $\alpha$ is the convolution of a Gaussian of radius $\sigma_E = \beta^{-1}$ with $\gamma(\omega) = \exp\left(-\beta \max
\left(\omega + \frac{\beta\sigma_E^2}{2}, 0\right)\right)$, the assumption that $\beta\norm{\vH} = O(1)$ again yields that $\alpha$ evaluated at $\nu_{lm}$ is $\Omega(1)$. Indeed, within $O(\beta^{-1})$ of any value of $\nu_{lm}$, $\gamma(\omega)$ is $\Omega(1)$, and as a consequence $\alpha(\nu_{lm}) = \Omega(1)$. 
This yields a lower bound on $\lambda$ of $\min_{l \neq m} \frac{\Omega(n^{-1})}{O(n^{1})} = \Omega(1)$. 

Now, we lower bound the diagonal elements outside of the classical block. 
Since each such element is of the form

\[\mathbb{E}\left[\bra{\mathbf{m}}\mathcal{L}_d\ket{\mathbf{m}}\right] = -\frac{1}{2} \left( \frac{1}{n}\sum_{j}\alpha(\nu_{jm_2}) + \frac{1}{n}\sum_{ j}\alpha(\nu_{jm_1})\right)\]
and we have already established that each $\alpha$ term is $\Omega(1)$, so the resulting diagonal values are all $\Omega(1)$. We conclude that the spectral gap of the Lindbladian is $\Omega(1)$.

\subsubsection{Proof of Theorem \ref{thm:constgapbounded_intro}}\label{pf:constgapbounded}
We construct our Lindbladian by sampling $M = \Theta(\log(n))$ unnormalized jumps $\vA^a$ from the 1-design $\mathcal{D}(U(n))$ as in Definition~\ref{def:design}, each with a corresponding Lindbladian $\mathcal{L}_a$ (which has one jump $\vA^a$ along with its adjoint, normalized by 2). 
Then, we want to prove that with high probability, $\mathcal{L}_\beta = \frac{2}{M}\sum_{a=1}^{M/2} \mathcal{L}_a$, the Lindbladian with all $M$ of these jumps now normalized by the number of jumps, has spectral gap bounded by a constant. 

To prove the result, we shall make use of the matrix Bernstein's inequality for our concentration bound:

\begin{lem}[cf. \cite{tropp2015introduction}]\label{lem:matrixbernstein}
Say $X_1, \dots, X_{N}$ are independent random $d \times d$ Hermitian matrices, such that $\mathbb{E}[X_i] = 0$ and $\norm{X_i} \leq R$. 
Define $Y = \frac{1}{N}\sum_{i=1}^N X_i$, and say that $N\mathbb{E}[Y^2] \leq T$. Then $\Pr(\norm{Y} \geq t) \leq 2d \exp(-\frac{3}{2}\frac{Nt^2}{3T + Rt})$.
\end{lem}

Call $\delta \mathcal{L}_a = \mathcal{L}_a - \mathcal{L}_\mu$, where $\mathcal{L}_\mu = \mathbb{E}_{\vA \sim \mathcal{D}}[\mathcal{L}_\beta]$. 
Each of these operators has zero expectation.
We have that $\delta \mathcal{L} = \frac{2}{M}\sum_{a=1}^{M/2} \delta \mathcal{L}_a$ is precisely the discrepancy between our Lindbladian $\mathcal{L}_\beta = \frac{2}{M}\sum_{a=1}^{M/2}\mathcal{L}_a$ and the expected Lindbladian $\mathcal{L}_\mu$. 
We will apply Bernstein's inequality for $X_a = \delta \mathcal{L}_a$, $Y = \delta \mathcal{L}$, and $N = \frac{M}{2}$.
By Lemma~\ref{lem:operatornorm}, the CKG Lindbladian has operator norm $O(\log(\beta \norm{\vH}))$, which is $O(1)$ in this regime.
We denote this upper bound $\frac{R}{2}$. 
We can now verify the condition $N\mathbb{E}[Y^2] \leq T$ in the statement of Lemma~\ref{lem:matrixbernstein}, since we have that $\norm{\delta \mathcal{L}_a} \leq \norm{\mathcal{L}_a} + \norm{\mathcal{L}_\mu} \leq R = \Theta(1)$, and 

\[\frac{M}{2}\left \lVert\mathbb{E}\left[\left(\frac{\sum_{a=1}^{M/2} \delta \mathcal{L}_a}{M/2}\right)^2\right]\right \rVert \leq \frac{2}{M}\sum_{a=1}^{M/2} \norm{\delta \mathcal{L}_a}^2 \leq R^2.\]

Every operator that satisfies detailed balance is Hermitian in some fixed basis, and therefore each $\mathcal{L}_a$, as well as $\mathcal{L}_\mu$, can be considered Hermitian. 
By linearity, the same holds true for $\delta \mathcal{L}_a = \mathcal{L}_a - \mathcal{L}_\mu$.
The operators $\delta \mathcal{L}_a$ therefore satisfy the conditions of the matrix Bernstein inequality, and so their average $\delta \mathcal{L}$ satisfies

\[\Pr\left(\left\Vert \delta \mathcal{L} \right \Vert\geq t\right) \leq 2n^2\exp\left(-\frac{3}{4}\frac{Mt^2}{3R^2 + Rt}\right),\]
where the $n^2$ is due to an overhead of the dimension of the Lindbladian.

For any constant $t$, there exists an $M=\Theta(\log(n))$ such that the term inside the exponential is at least $3\log(n)$, since $R$ is a constant. 
The probability that $\norm{\delta \mathcal{L}} \leq t$ is then arbitrarily close to 1 for sufficiently large $n$ and choice of $M = \Theta(\log(n))$. 
By Lemma \ref{lem:boundedaveragelindbladian}, $\mathcal{L}_\mu$ has a constant spectral gap bounded below by some $C$. 
By Weyl's theorem, the eigenvalues of $\mathcal{L}_\mu + \delta \mathcal{L}$ may differ by at most $t$ from those of $\mathcal{L}_\mu$. 
Choosing $t \leq \frac{C}{2}$, it follows that $\mathcal{L}_a$, with any constant probability, has constant spectral gap. 

\subsection{Unbounded Degree Systems with Haar Random Jumps}\label{sec:UnboundedProof}
\subsubsection{Proof of Lemma \ref{lem:averagelindbladian}}
We follow the proof of Lemma \ref{lem:boundedaveragelindbladian}.
The expected Lindbladian takes the following form:
\begin{align*}
\mathbb{E}\left[\bra{\mathbf{l}}\mathcal{L}_t\ket{\mathbf{m}}\right] &= \frac{\delta_{l_1l_2}\delta_{m_1m_2}}{n}\alpha(\nu_{l_1m_1}),\\
\mathbb{E}\left[\bra{\mathbf{l}}\mathcal{L}_d\ket{\mathbf{m}}\right] &= -\frac{1}{2} \left( \frac{\delta_{l_1m_1}\delta_{l_2m_2}}{n}\sum_{j}\alpha(\nu_{jm_2}) + \frac{\delta_{l_1m_1}\delta_{l_2m_2}}{n}\sum_{ j}\alpha(\nu_{jm_1})\right),
\\
\mathbb{E}\left[\bra{\mathbf{l}}\mathcal{L}_c\ket{\mathbf{m}}\right] &=  0.
\end{align*}
The final Lindbladian $\mathcal{L}_\mu$ is diagonal except for a ``classical block'' $\mathcal{L}^0$ of indices $\ket{(m, m)}$, whose off-diagonal terms are populated by the elements of $\mathcal{L}_t$. 
The spectral gap of this Lindbladian is therefore the minimum of the values along the diagonal, which are all positive, and the spectral gap of the classical block. 

As in Lemma \ref{lem:boundedaveragelindbladian}, we may obtain the following bound on the classical block:
\begin{align}
\lambda \geq \min_{l\neq m} \frac{\alpha(\nu_{lm})n^{-1}}{\rho_{ll}}.
\end{align}
Now, note that $\rho_{ll} = -n\exp(-\beta E_l)\left(\frac{1}{n}\sum_j \exp(-\beta E_j)\right)$, so the above expression simplifies 

\[\lambda \geq \min_{l \neq m}\frac{\alpha(\nu_{lm})\left(\frac{1}{n}\sum_j \exp(-\beta E_j)\right)}{\exp(-\beta E_l)}.\]
By assumption, $\delta(n)$ of the eigenvalues lie within $C\beta^{-1}$ of $E_\text{min}$ for some constant $C$. 
We may therefore bound
\[\left(\frac{1}{n}\sum_j \exp(-\beta E_j)\right) \geq \delta(n)\exp(-\beta E_\text{min} - C).\]
Moreover, $\alpha$ is a convolution of a Gaussian with a monotone function $\lambda$.
So, $\alpha(\nu_{lm}) \geq \frac{1}{2}\gamma(\nu_{lm}) = \frac{1}{2}\exp\left(-\beta \max\left(\omega + \beta \frac{\sigma_E^2}{2}, 0\right)\right)$, which for $\sigma_E = \beta^{-1}$ is easily seen to be within a constant $D$ of $\exp(-\beta\max(\omega, 0))$. 

So we write
\begin{align*}
\lambda &\geq D\exp(-C)\min_{l \neq m} \frac{\delta(n)\exp(-\beta\max(\nu_{lm}, 0))\exp(-\beta E_\text{min})}{\exp(-\beta E_l)}\\
&\geq D\exp(-C)\min_{l \neq m} \frac{\delta(n)\exp(-\beta E_\text{min})}{\exp(-\beta \max(E_l, E_m))}\\
&= \Omega(\delta(n)).
\end{align*}
We now lower bound the diagonal elements outside of the classical block. 
Each such element is of the form
\[\mathbb{E}\left[\bra{\mathbf{m}}\mathcal{L}_d\ket{\mathbf{m}}\right] = -\frac{1}{2} \left( \frac{1}{n}\sum_{j}\alpha(\nu_{jm_2}) + \frac{1}{n}\sum_{ j}\alpha(\nu_{jm_1})\right).\]
By assumption, we have that $\delta(n)$ of eigenvalues are within $O(\beta^{-1})$ of $\lambda_\text{min}$. 
For these eigenvalues, $\nu_{jm_1}$ and $\nu_{jm_2}$ are at most $O(\beta^{-1})$, and therefore $\alpha$ is bounded below. 
The above sum is therefore $\frac{1}{n}\Omega(\delta(n)n) = \Omega(\delta(n))$, as desired.

\subsubsection{Proof of Theorem~\ref{thm:gapunbounded_intro}}\label{pf:gapunbounded_intro}
We follow the proof of Theorem $\ref{thm:constgapbounded_intro}$. 
Defining once again $\mathcal{L}_a$ to be the Lindbladian with one jump operator $\vA^a$ and its adjoint (normalized by 2), and defining $\delta \mathcal{L}_a = \mathcal{L}_a - \mathcal{L}_\mu$, we can apply the matrix Bernstein's inequality to obtain
\[\Pr\left(\left\Vert \delta \mathcal{L} \right \Vert\geq t\right) \leq n^2\exp\left(-\frac{3}{4}\frac{Mt^2}{3R^2 + Rt}\right),\]
since all the conditions of the inequality are again satisfied. Here $\delta \mathcal{L} = \frac{2}{M}\sum_{a=1}^{M/2} \delta \mathcal{L}_a$, so that $\mathcal{L}_\beta = \mathcal{L}_\mu + \delta \mathcal{L}$. $R = O(\log(\beta \norm{\vH}))$, since $\norm{\delta \mathcal{L}_a} \leq \norm{\mathcal{L}_a} + \norm{\mathcal{L}_\mu} = O(\log(\beta \norm{\vH}))$. This latter equality arises from the operator norm bound of the CKG Lindbladian by Lemma~\ref{lem:operatornorm}.

By Lemma \ref{lem:boundedaveragelindbladian}, $\mathcal{L}_\mu$ has a constant spectral gap bounded below by some $C\delta(n)$. Selecting $t$ to be $\frac{C}{2}\delta(n)$, there exists an $M = \Theta(\delta(n)^{-2}\log(\beta\norm{ \vH})^2\log(n))$ for which the value inside the exponential is at least $3\log(n)$. 
The probability that $\norm{\delta \mathcal{L}} \leq t$ is therefore above any constant probability for sufficiently large $n$. 
By Weyl's theorem, the eigenvalues of $\mathcal{L}_\mu + \delta \mathcal{L}$ may differ by at most $t \leq \frac{C}{2}\delta(n)$. 
It follows that $\mathcal{L}_a$, with any constant probability, has spectral gap $\Omega(\delta(n))$. 

\subsection{Unbounded Degree System with Local Jumps}
In this section, we prove the spectral gap bounds on the Lindbladian for a a hypercube, where using well-chosen local jumps achieves polylogarithmic time complexity in the number of vertices (polynomial in the number of qubits), an exponential improvement over naive 1-design jumps. 

\subsubsection{Proof of Theorem \ref{thm:hypercubelocalgap}}
The proof is closely related to the spectral gap bound of the non-interacting Hamiltonian in \cite{consttempquantadv}. 
We first prove the following lemma:
\begin{lem}
The CKG Lindbladian $\mathcal{L}_\beta$ with Hamiltonian $H_1 \otimes I + I \otimes H_2$ and jump operator $\vA \otimes I$ can be decomposed as $\mathcal{L}_1 \otimes I$ (i.e., it only acts on the first space).
\end{lem}
\begin{proof}
We will rebuild the Lindbladian from its definition to obtain the result.
Firstly, the operator Fourier transform of the jump operators yields
\begin{align*}\hat{\vA}^a(\omega) &= \frac{1}{\sqrt{2\pi}}\int_{-\infty}^\infty e^{iHt }(\vA^a\otimes I) e^{-iHt}e^{-i\omega t}f(t)dt\\
&= \left(\frac{1}{\sqrt{2\pi}}\int_{-\infty}^\infty e^{iH_1t }\vA^a e^{-iH_1t}e^{-i\omega t}f(t)dt\right) \otimes I.
\end{align*}
Then,

\[\mathcal{L}_\beta[\cdot] = -i[\vB, \cdot] + \sum_{a \in [M]} \int_{-\infty}^\infty \gamma(\omega) \left(\hat{\vA}^a(\omega) (\cdot) \hat{\vA}^a(\omega)^\dag - \frac{1}{2}\{\hat{\vA}^a(\omega)^\dag\hat{\vA}^a(\omega), \cdot\}\right)d\omega.\]
Since the dissipative part of the Lindbladian only acts on the state with $\hat{\vA}^a(\omega)$, it can only act on the the first space. 
We therefore must only prove that $-i[\vB, \cdot]$ acts only on the first space. 
In \cite{chen2023efficient} it is established that there is a unique such operator $-i[\vB, \cdot]$ such that the sum satisfies detailed balance with repect to $\rho_\beta$.
Since $\vH$ is separable, $\rho_\beta = \rho_\beta^1 \otimes \rho_\beta^2$. 
The form of detailed balance is
\[\mathcal{L}_\beta[\cdot] = \vrho_\beta^{1/2}\mathcal{L}_\beta^\dag[\vrho_\beta^{-1/2} (\cdot) \vrho_\beta^{-1/2}]\vrho_\beta^{1/2}.\]
As $\rho_\beta$ is separable, satisfying detailed balance for two distinct operators $\mathcal{L}_1$ for $\rho_\beta^1$ and $\mathcal{L}_2$ for $\rho_\beta^2$ is sufficient to satisfy the detailed balance for $\mathcal{L}_\beta = \mathcal{L}_1 \otimes \mathcal{L}_2$. 

We may therefore find some $-i[\vB_1, \cdot]$ such that the sum with the dissipative Lindbladian satisfies detailed balance with repect to $\rho_\beta^1$. 
Taking the tensor product with $\mathcal{L}_2 = I$ yields a Lindbladian that satisfies detailed balance with respect to $\rho_\beta$, and therefore by uniqueness $\vB = \vB_1 \otimes \vI$. 
This completes the proof. 
\end{proof}
Treating the vertices of a hypercube with $2^d$ vertices as bitstrings of length $d$, with an edge between two bitstrings of Hamming distance 1, the Hamiltonian can be represented as $H = \sum_{i = 1}^d X_i$. 
Choosing the jump operators to be $\frac{1}{\sqrt{d}}Z_a$, the resulting CKG Lindbladian can be written as the sum of $d$ Lindbladians for each jump operator. 
Since these jump operators only act on one index, by the lemma above we may consider $\mathcal{L}_\beta$ as a sum $\mathcal{L}_\beta = \sum_{i=1}^d \mathcal{L}_i$, where each acts on a 4-dimensional space. 
It follows that the spectral gap of $\mathcal{L}_\beta = \min_i \lambda_{\text{gap}}(\mathcal{L}_i)$. 

We may now explicitly write $\mathcal{L}_i$ in the energy basis ($\ketbra{0}{0}, \ketbra{1}{1}, \ketbra{0}{1}, \ketbra{1}{0}$), using the previous formula \eqref{eq:graphlindbladian}:
\[\frac{1}{2d}\begin{bmatrix}
-2\alpha(-2) & 2\alpha(2) & 0 & 0 \\
2\alpha(-2) & -2\alpha(2) & 0 & 0 \\
0 & 0 & -\alpha(-2)-\alpha(2) & 2\theta(2, -2) \\
0 & 0 & 2\theta(2, -2) & -\alpha(-2)-\alpha(2)
\end{bmatrix}.\]

The first block has rank 1 and the trace is $-\frac{\alpha(-2)+\alpha(2)}{d}$, giving the eigenvalues as $0$ and $-\frac{\alpha(-2) + \alpha(2)}{2d}$. 
Moreover, $-\alpha(-2) - \alpha(2) + 2\theta(2, -2) < 0$.
We conclude that the spectral gap is $\Omega(\frac{1}{d})$ for a sufficiently small constant $\sigma_E$ given that $\theta(2, -2) = \alpha(0)\exp(-\frac{2}{\sigma_E^2})$.

Since this lower bound holds for each index, this completes the proof that the spectral gap is $\Omega(\frac{1}{d})$. 
\bibliographystyle{alphaUrlePrint.bst}
\bibliography{ref,qc_gily}

\end{document}